\documentclass[11pt]{article}
\date{\today}

\usepackage[margin=1in]{geometry}
\usepackage[utf8]{inputenc}
\usepackage{amsthm}
\usepackage{amsmath}
\usepackage{amsfonts}
\usepackage{amssymb}
\usepackage[colorlinks]{hyperref}
\hypersetup{
  colorlinks,
  citecolor=Blue,
  linkcolor=Blue,
  urlcolor=Blue}
\usepackage{cleveref}
\usepackage[dvipsnames]{xcolor}
\usepackage{fullpage}
\usepackage{algorithm,caption}
\usepackage{algpseudocode}
\usepackage{xspace}
\usepackage{multirow}
\usepackage{diagbox}
\usepackage{thm-restate}
\usepackage{multirow}
\usepackage{array}
\usepackage{siunitx}
\usepackage{booktabs}

\usepackage{titlesec}

\algnewcommand\algorithmicinput{\textbf{Parameters:}}
\algnewcommand\Params{\item[\algorithmicinput]}
\usepackage{derivative}
\usepackage{mathtools}
\usepackage{bbm}
\usepackage{tikz}
\usepackage{forest}

\newtheorem{theorem}{Theorem}[section]

\newtheorem{lemma}[theorem]{Lemma}

\newtheorem{definition}[theorem]{Definition}

\newtheorem{claim}[theorem]{Claim}

\newtheorem{insight}{Insight}

\newcommand{\maxcsps} {\textsc{Max-CSP}s\xspace}
\newcommand{\maxkcsps} {\textsc{Max-$k$-CSP}s\xspace}

\newcommand{\maxcsp} {\textsc{Max-CSP}\xspace}
\newcommand{\ug} {\textsc{Unique Games}\xspace}
\newcommand{\mincsp}{\textsc{Min-CSP}\xspace}
\newcommand{\minkcsp}{\textsc{Min-\mbox{$k$}-CSP}\xspace}
\newcommand{\minkcsps}{\textsc{Min-\mbox{$k$}-CSP}s\xspace}
\newcommand{\mincsps}{\textsc{Min-CSP}s\xspace}

\newcommand{\minuncut}{\textsc{Min-Uncut}\xspace}
\newcommand{\mintwocnf}{\textsc{Min-$2$-CNF-Deletion}\xspace}
\newcommand{\mintwo}{\textsc{Min-$2$-SAT}\xspace}
\newcommand{\maxtwo}{\textsc{Max-$2$-SAT}\xspace}

\newcommand{\maxnaethree}{\textsc{Max-NAE-\mbox{$3$}-SAT}\xspace}
\newcommand{\minnaethree}{\textsc{Min-NAE-\mbox{$3$}-SAT}\xspace}
\newcommand{\minthree}{\textsc{Min-\mbox{$3$}-SAT}\xspace}
\newcommand{\nae}{\textsc{NAE\mbox{-}$3$-SAT}\xspace}

\newcommand{\nsat}[1]{$#1$-\textsc{SAT}\xspace}
\newcommand{\ksat}{\nsat{k}}

\newcommand{\kcsp}{\ncsp{k}}
\newcommand{\kcsps}{$k$-\textsc{CSP}s\xspace}

\newcommand{\threesat}{\nsat{3}}
\newcommand{\twosat}{\nsat{2}}
\newcommand{\klin}{$k$-\textsc{Lin}}
\newcommand{\kand}{$k$-\textsc{And}}

\newcommand{\csp}{\textsc{CSP}\xspace}
\newcommand{\csps}{\textsc{CSP}s\xspace}
\newcommand{\ncsp}[1]{$#1$-\textsc{CSP}\xspace}

\newcommand{\np}{\textsf{NP}\xspace}
\newcommand{\nph}{\textsf{NP}\mbox{-}hard\xspace}
\newcommand{\nphn}{\textsf{NP}\mbox{-}hardness\xspace}
\newcommand{\bpp}{\textsf{BPP}\xspace}

\newcommand{\poly}{\text{\rm poly}}
\newcommand{\polylog}{\text{\rm  polylog}}

\newcommand{\lp}{\ensuremath{\mathrm{LP}}}

\newcommand{\opt}{\ensuremath{\mathrm{OPT}}}
\newcommand{\eps}{\varepsilon}

\newcommand{\val}{\ensuremath{\mathrm{val}}}
\newcommand{\supp}{\ensuremath{\mathrm{supp}}}
\newcommand{\roundalgo}{\textsc{round-pd}}
\newcommand{\kl}{\textsc{KL}}
\newcommand{\dkl}{D_\kl}
\renewcommand{\hat}{\widehat}
\renewcommand{\tilde}{\widetilde}
\DeclareMathOperator*{\ex}{\mathbb{E}}
\newcommand{\fix}{\ensuremath{\mathrm{fix}}}
\newcommand{\unsat}{\ensuremath{\mathrm{unsat}}}
\newcommand{\rate}{\ensuremath{\mathrm{rate}}}
\newcommand{\eul}{\ensuremath{\mathrm{e}}}

\newcommand{\id}{\ensuremath{\mathrm{Id}}}
\newcommand{\inv}{\ensuremath{\overline{\mathrm{Id}}}}

\newcommand{\vecc}{\ensuremath{\mathbf{c}}}

\newcommand{\vecs}{\ensuremath{\mathbf{s}}}
\newcommand{\vect}{\ensuremath{\mathbf{t}}}

\newcommand{\cC}{\ensuremath{\vecc}}
\newcommand{\calC}{\ensuremath{\mathcal{C}}}

\newcommand{\cD}{\ensuremath{\mathcal{D}}}
\newcommand{\E}{\ensuremath{\mathbb{E}}}

\newcommand{\rF}{\ensuremath{\mathrm{F}}}

\newcommand{\cI}{\ensuremath{\mathcal{I}}}

\newcommand{\N}{\ensuremath{\mathbb{N}}}

\newcommand{\cP}{\ensuremath{\mathcal{P}}}

\newcommand{\R}{\ensuremath{\mathbb{R}}}

\newcommand{\cS}{\ensuremath{\vecs}}

\newcommand{\cT}{\ensuremath{\vect}}

\newcommand{\rU}{\ensuremath{\mathrm{U}}}

\title{Min-CSPs on Complete Instances II:\\
	Polylogarithmic Approximation for Min-NAE-3-SAT}
\author{Aditya Anand\thanks{University of Michigan. Email: \url{adanand@umich.edu}.}
	\and Euiwoong Lee\thanks{University of Michigan. Email: \url{euiwoong@umich.edu}.
		Supported in part by NSF grant CCF-2236669 and Google.} \and Davide Mazzali\thanks{EPFL. Email: \url{davide.mazzali@epfl.ch}.} \and Amatya Sharma\thanks{University of Michigan. Email: \url{amatya@umich.edu}.}}
\date{}

\begin{document}
	\pagenumbering{gobble}
	
	\maketitle
	
	\begin{abstract}

		This paper studies complete $k$-Constraint Satisfaction Problems (CSPs), where an $n$-variable instance has exactly one nontrivial constraint for each subset of $k$ variables, i.e., it has $\binom{n}{k}$ constraints. A recent work started a systematic study of complete $k$-CSPs [Anand, Lee, Sharma, SODA'25], and showed a quasi-polynomial time algorithm that decides if there is an assignment satisfying all the constraints of any complete Boolean-alphabet $k$-CSP, algorithmically separating complete instances from dense instances.

		The tractability of this decision problem is necessary for any nontrivial (multiplicative) approximation for the minimization version, whose goal is to minimize the number of violated constraints. The same paper raised the question of whether it is possible to obtain nontrivial approximation algorithms for complete \textsc{Min}-$k$-CSPs with $k \geq 3$.

		In this work, we make progress in this direction and show a quasi-polynomial time $\polylog(n)$-approximation to \textsc{Min-NAE-3-SAT} on complete instances, which asks to minimize the number of $3$-clauses where all the three literals equal the same bit. To the best of our knowledge, this is the first known example of a CSP whose decision version is \textsf{NP-Hard} in general (and dense) instances 
		while admitting a $\polylog(n)$-approximation in complete instances. Our algorithm presents a new iterative framework for rounding a solution from the Sherali-Adams hierarchy, where each iteration interleaves the two well-known rounding tools: the {\em conditioning} procedure, in order to ``almost fix'' many variables, and the {\em thresholding} procedure, in order to ``completely fix'' them.
		
		Finally, we improve the running time of the decision algorithms of Anand, Lee, and Sharma and show a simple algorithm that decides any complete Boolean-alphabet $k$-CSP in polynomial time.
	\end{abstract}
	
	\newpage{}
	
	\tableofcontents
	
	\newpage{}
	
	\pagenumbering{arabic}
	\setcounter{page}{1}
	
	\section{Introduction}

Constraint Satisfaction Problems (\csps) provide a unified framework for expressing a wide range of combinatorial problems, including SAT, Graph Coloring, and Integer Programming. 
A \csp consists of a set $n$ of variables that must be assigned values from a given alphabet, subject to a collection of $m$ constraints. The goal is typically to determine whether there exists an assignment that satisfies all constraints (decision \csps) or to optimize some measure of constraint satisfaction (optimization \csps).

A major example of optimization \csps is \maxcsps, where the objective is to find an assignment that maximizes the number of satisfied constraints. This class of problems has been extensively studied and it is known to admit various approximation algorithms, including the (conditionally) optimal approximability of fundamental problems such as \textsc{Max}-$3$-\textsc{LIN}, \textsc{Max}-$3$-\textsc{SAT}, \textsc{Max}-\textsc{Cut}, and \ug~\cite{Hastad01, Khot02, KKMO07}. However, significantly less is known for their minimization counterpart, namely \mincsps, which aim to minimize the number of violated constraints.

In many cases, the minimization versions of \csps are inherently harder to approximate than the maximization objective. For example, while \maxtwo admits a tight $\alpha_{\text{LLZ}}\approx 0.94016567$ approximation algorithm modulo the Unique Games Conjecture (UGC) \cite{lewin2002improved, brakensiek2024tight, Khot02}, the best-known approximation guarantee for \mintwo (also known as \mintwocnf) is an $O(\sqrt{\log n})$-approximation \cite{ACMM05} along with a hardness of $\omega(1)$ under the UGC \cite{Khot02}. Another example is that of \textsc{Not-All-Equal}-\threesat\ (\nae), where a clause (consisting of three literals) is satisfied if and only if not all the three literals evaluate to the same Boolean value. The maximization version, \maxnaethree admits a tight approximation factor of $\approx 0.9089$ modulo the UGC \cite{brakensiek2021mysteries}. In contrast, the minimization version \minnaethree cannot even admit any finite approximation in polynomial time, simply from the \nphn of the decision version \nae~\cite{sch78}.

Understanding the approximability of Min-CSPs remains a major challenge in computational complexity and approximation algorithms. On the brighter side, \cite{khanna2001approximability} proved that the optimal approximation ratio takes one of the values in $\{ 1, O(1), \polylog(n), \poly(n), \infty \}$ for any \mincsp with the Boolean alphabet (here in this paper, we will talk about \csps on the Boolean alphabet, unless specified), based on some structural classification.

Apart from the \emph{general} instances, CSPs often exhibit significantly different behavior when structural requirements are imposed on their instances. The structure of a CSP refers to assumptions on how constraints are distributed across the instance. A \ncsp{k} is defined on $n$ variables, where each constraint involves exactly $k$ variables, and the constraint structure can naturally be modeled as a $k$-uniform hypergraph, where variables correspond to vertices and each constraint corresponds to a hyperedge. Structural assumptions on \csps often translate into density conditions on the corresponding hypergraph. Two particularly important structured settings are (1) \emph{dense instances}, where the number of constraints (i.e. hyperedges) is $\Omega(n^k)$, and (2) \emph{complete instances}, where the number of constraints is exactly $\binom{n}{k}$.

\maxcsps on dense instances have been extensively studied, with powerful algorithmic techniques yielding strong approximation guarantees. In fact, for every \maxcsp on dense instances, a Polynomial-Time Approximation Scheme (PTAS) is known, achievable through any of the three major algorithmic frameworks: {\em random sampling}~\cite{arora1995polynomial, 
bazgan2003polynomial, alon2003random, de2005tensor, mathieu2008yet, KS09, barak2011subsampling, yaroslavtsev2014going, manurangsi2015approximating, fotakis2016sub}, {\em convex hierarchies}~\cite{de2007linear, arora2008unique, BRS11, GS11, yz14, alev2019approximating, jeronimo2020unique, bafna2021playing}, and {\em regularity lemmas}~\cite{frieze1996regularity, coja2010efficient, oveis2013new, jeronimo2021near}. However, in contrast, \mincsps on dense instances remain far less explored. Known results exist only for specific problems, such as $O(1)$-approximation algorithms for \ug and \minuncut~\cite{bazgan2003polynomial, KS09, GS11, meot2023voting}, and a PTAS for fragile \mincsps \cite{KS09}, where a \csp is fragile when changing the value of a variable always flips a clause containing it from satisfied to unsatisfied. Despite these advances, a general framework for tackling \mincsps in dense settings is still lacking, making it a compelling direction for further study.

Building on the study of structured CSPs, \cite{anand2025min} introduced complete instances as an extreme case of structure, where every possible $k$-set of variables forms a constraint. The primary motivation for the study of complete instances comes from their connections to machine learning and data science, their role in unifying algorithmic techniques for dense \maxcsps, and the structural insights they provide into \csps (see \cite{anand2025min} for details). They give a constant-factor polynomial time approximation for \mintwo on complete instances, which contrasts to the $\omega(1)$-hardness on dense instances, 
and quasi-polynomial (specifically, $n^{O(\log n)}$) time algorithms for the decision versions of \ksat\ and \kcsp for all constants $k$. They also give a polynomial time algorithm for deciding \nae on complete instances. On the hardness side, they prove that there is no polynomial time algorithm for exact optimization (which does not distinguish the maximization and minimization versions) of \nae, \ksat, \klin, \kand, and \ug even on complete instances unless $\np \subseteq \bpp$.

\paragraph{Our Contribution.} We prove two main results in our work: (1) a quasi-polynomial time $O(\log^6 n)$ approximation for \minnaethree on complete instances, and, (2) a polynomial time algorithm for the decision version of \kcsp for all constants $k$.

We first present a quasi-polynomial (i.e., $n^{\polylog(n)}$) time  $\polylog(n)$-approximation algorithm for \minnaethree on complete instances.
To the best of our knowledge, it is the first known example of a CSP whose decision version is \nph in general instances (and dense instances too, see Claim~\ref{claim:hardness-dense})
while admitting a $\polylog(n)$-approximation in complete instances.

\begin{theorem}\label{thm:nae}
    There is an algorithm running in time $n^{\log ^\kappa n}$ for some constant $\kappa > 0$ that gives an $O(\log^6 n)$-approximation for \minnaethree on complete instances. 
    Furthermore, the integrality gap of the degree-$O(\log ^\kappa n)$ Sherali-Adams relaxation is at most $O(\log^6 n)$.
\end{theorem}

\noindent
Beyond addressing this specific question, our result also strengthens one of the main motivations for studying \mincsps on complete instances: understanding whether a combination of algorithmic techniques (random sampling, convex hierarchies, and regularity lemmas) can improve approximability results. As stated earlier, while each of these techniques independently yields PTASes for \maxcsps on dense/complete instances, their effectiveness for \mincsps is much less understood. Our algorithm presents a new iterative framework for rounding a solution from the $\polylog(n)$-round Sherali-Adams hierarchy, where each iteration interleaves the two well-known rounding tools: (1) the {\em conditioning} procedure, which identifies and conditions on a small set of variables in order to almost fix a constant fraction of variables and (2) the {\em thresholding} procedure, in order to completely fix them. We fix a constant fraction of variables in each iteration (for a total of $O(\log n)$ iterations) while ensuring that the value of the LP remains within a $\polylog(n)$ factor of the optimal value by the end of all the iterations.

Secondly, we give polynomial time algorithms for the decision version of \kcsps (for every constant $k$), improving over the quasi-polynomial time algorithm of~\cite{anand2025min}. 

\begin{theorem}\label{thm:kcsp}
    For every $k \ge 2$, there is a polynomial time algorithm that decides whether a complete instance of \kcsp is satisfiable or not.
\end{theorem}

\noindent
Our algorithm is remarkably simple. It relies on Lemma $3.1$ of  \cite{anand2025min}, which shows that the number of satisfying assignments for a complete instance of a \kcsp is at most $O(n^{k-1})$, based on a VC-dimension argument. The algorithm first arbitrarily orders the variables $v_1, \ldots, v_n$, and then proceeds by iteratively maintaining all satisfying assignments for the partial (complete) instance induced on the first $i$ variables. Since the number of satisfying assignments for a complete instance with $i$ variables is at most $O(i^{k-1})$, we can efficiently keep track and update these solutions in polynomial time, leading to our second main result.

\paragraph{Open Questions.}  Our work establishes the first nontrivial approximation result for a \mincsp, \minnaethree on complete instances, whose decision version is \nph on general and dense instances. \nae remains hard even when every variable appears in $\Omega(n^2)$ clauses (\Cref{claim:hardness-dense}). 
However, our techniques do not currently extend to other problems including \minthree and \textsc{Min-NAE-4-SAT}.

This raises the open question regarding both approximation algorithms and hardness results for \minkcsps. While exact optimization is known to be hard, proving inapproximability beyond this --- such as APX-hardness --- remains an important direction. A natural first candidate is \minthree, for which the existence of efficient quasi-polynomial time approximation algorithms or even hardness results on complete instances, remains unresolved.

\paragraph{Organization.} \Cref{sec:overview} provides an overview of the quasi-polynomial time $O(\log^6 n)$ approximation for \minnaethree, i.e., a proof outline for \Cref{thm:nae}. \Cref{sec:prelims} establishes the notations and preliminaries required throughout the paper. \Cref{sec:roundingalgo} shows the proof of the rounding algorithm, \Cref{alg:roundingalgo}, which proves \Cref{thm:nae} by rounding the solution of the Sherali-Adams LP. \Cref{sec:kcsp} presents the polynomial time algorithm for the decision \kcsps (\Cref{thm:kcsp}). 
	\section{Technical overview}
\label{sec:overview}
Convex programming hierarchies have proven to be instrumental in the design of approximation algorithms of \maxcsps~{\cite{de2007linear, arora2008unique, BRS11, GS11, yz14, alev2019approximating, jeronimo2020unique, bafna2021playing}} and, limited to a few special cases, \mincsps~{\cite{CCLLNV24}}. Therefore, it is natural to attempt to approximate \minnaethree on complete instances using these approaches. We prove \Cref{thm:nae} by designing an algorithm that first runs the degree-$d$ Sherali-Adams relaxation of \minnaethree, and then rounds the fractional solution to a Boolean assignment.

Informally, a fractional solution $\mu$ to such a relaxation associates each set $S$ of at most $d$ variables with distributions $\mu_S$ over assignments $\alpha \in \{0,1\}^S$ to the variables in $S$. The crucial property is that these distributions are \textit{locally consistent}: for any two variable sets $S,T$ of size at most $d$, the local distributions $\mu_S$ and $\mu_T$ agree on the probability of each assignment to $S\cap T$. For this reason, such a solution $\mu$ is called a  \textit{pseudodistribution}.

The objective of the relaxation is to minimize the average over constraints of the probability that they are violated by an assignment sampled from their respective local distribution.  In particular, the objective value of an optimal pseudodistribution $\mu^*$ is at most the fraction of violated constraints in an optimal Boolean assignment.

As solving a degree-$d$ Sherali-Adams relaxation amounts to solving a linear program with $n^{O(d)}$ variables, we would like to keep $d$ as small as possible. In our context, we will set $d=\polylog(n)$.

Having put this notation in place, we now give an overview of how the algorithm from \Cref{thm:nae} rounds these fractional solutions. Specifically, in \Cref{subsec:techoverview_pd} we highlight some of the structural properties that our problem instance impose on the pseudodistributions, and in \Cref{subsec:techoverview_algo} we discuss how to exploit these properties to design a rounding algorithm.

\subsection{Pseudodistributions and complete NAE-3-SAT}
\label{subsec:techoverview_pd}
When rounding pseudodistributions obtained from the Sherali-Adams or Sum-of-Squares hierarchies, \textit{conditioning} is perhaps the main hammer at our disposal. This technique, introduced by \cite{BRS11, GS11} and adopted ubiquitously to approximate optimization \csps, is usually implemented as follows: first, we condition on a subset of $r=O(1/\epsilon^2)$ variables to reduce the average correlations among the groups of $t$ variables to below $O(2^t\epsilon)$; then, we perform independent rounding. This scheme is particularly effective in the context of dense \maxkcsps, where we capitalize upon the fact that optimal assignment satisfies at least an $\Omega(1)$ fraction of constraints.
Hence, we can afford to additively lose a total variation distance of $O(2^k \epsilon)=O(\epsilon)$ between the pseudodistribution and independent rounding in each constraint, and still obtain a $(1-O(\epsilon))$-approximation.

\paragraph{Obstacle: we cannot afford additive error.}In the context of \minkcsps, the objective function is the fraction of violated constraints. Hence, for a \minkcsp instance $\cI$, we define $\opt$ to be the minimum over all the variable assignments $\alpha$ of the fraction of constraints in $\cI$ unsatisfied by $\alpha$. With such an objective, the trick discussed above ceases to work: a highly satisfiable instance has a very small optimal value, possibly even $\opt \le O(1/n^{k})$. Due to this eventuality, we cannot afford to additively lose an $\epsilon$ term for each constraint: it would result in an objective value of, say, $\epsilon+\opt \gtrsim n \cdot \opt$ unless $\epsilon \ll 1/n^{k-1}$. Unfortunately, the degree $d$ of the pseudodistribution needs to be larger than the number $r=O(1/\epsilon^2)$ of variables we condition on, so setting $\epsilon \ll 1/n^{k-1}$ is prohibitive.

\paragraph{Advantage: $\polylog(n)$ loss in time and approximation.} On the flip side of the above discussion, one can deduce that if $\opt \ge 1/\polylog(n)$ we can in fact afford to lose an additive $O(\epsilon)$ term for $\epsilon=1/\polylog(n)$ and still run in $n^{O(r)}=n^{\text{polylog}(n)}$ time. Even more so, if $\opt \ge 1/\polylog(n)$ then we always get a $\polylog(n)$-approximation, no matter what solution we output (since the value we obtain is always at most $1$). We can therefore restrain ourselves to consider instances with small optimal value.

\begin{insight}
\label{ins:unsat}
    Without loss of generality, we can consider only instances $\cI$ with  $\opt \le 1/\polylog(n)$.
\end{insight}

\paragraph{Advantage: the instance is complete.} Let $\mu$ be an optimal pseudodistribution for our instance $\cI$, and let us denote by $\delta = \val(\mu)$ the fractional objective value achieved by $\mu$ on $\cI$. Thanks to \Cref{ins:unsat}, we can assume $\delta \le \opt \le 1/\polylog(n)$. Recalling that we are in the realm of complete instances, a simple averaging yields the following observation.

\begin{insight}
\label{ins:unsatabs}
    For at least $\binom{n}{3}/2$ of the triples $\{u,v,w\}$, we have that the constraint $P_{\{u,v,w\}}(\alpha)$ is unsatisfied with probability at most $2\delta \le 1/\polylog(n) $ over $\alpha \sim \mu_{\{u,v,w\}}$.
\end{insight}
\noindent
We remark that a version of \Cref{ins:unsatabs} remains true for dense --- but not necessarily complete --- instances up to an $\Omega(1)$ loss in the lower bound. However, being dense is not a ``hereditary property'', while being complete is a ``hereditary property'': any variable-induced sub-instance of a complete instance is also complete and in particular dense, while this might not be the case for dense instances. Thus, assuming that $\cI$ is complete allows us to benefit from \Cref{ins:unsatabs} locally to any variable-induced sub-instance. As we will see, this is instrumental to our algorithm.

\paragraph{Advantage: the structure of \nae constraints.} Consider any constraint $P_{\{u,v,w\}}$ that is unsatisfied with probability at most $\delta'$ over $\mu_{\{u,v,w\}}$, for some $\delta' \in [0,1]$. By \Cref{ins:unsatabs}, we can assume that $\delta'\le 2\delta \le 1/\polylog(n)$. Then one can see that, among the six satisfying assignments of $P_{\{u,v,w\}}$, there is one that retains probability mass at least $(1-1/\polylog(n))/6\ge 1/7$ in $\mu_{\{u,v,w\}}$. Call this satisfying assignment $(\alpha^*_u,\alpha^*_v,\alpha^*_w)$. Assuming for simplicity that the literals of $P_{\{u,v,w\}}$ are all positive, we conclude that exactly two of $\alpha^*_u,\alpha^*_v,\alpha^*_w$ must be equal. By symmetry, we can assume that $\alpha^*_v=\alpha^*_w=0$ and $\alpha^*_u=1$. We now consider the pseudodistribution $\tilde{\mu}=\mu|_{(v,w)\leftarrow (\beta_v,\beta_w)}$ obtained by conditioning on the variables $v,w$ taking the random value $(\beta_v,\beta_w) \sim \mu_{v,w}$. We can observe that $(\beta_v,\beta_w)=(\alpha^*_v,\alpha^*_w)=(0,0)$ with probability at least $1/7$. Therefore, if such an event occurs, we have $\tilde{\mu}_u(0) \le 14\delta$, as one can deduce from \Cref{tab:fixing}. In this case, we say that $u$ becomes $O(\delta)$-fixed in $\tilde{\mu}$.

\begin{table}[H]
\begin{center}
\hspace{4.5em}
        \begin{tabular}{|@{}c cc@{}| l@{}}
            \cline{1-3}
            \, $v$ & $w$ & $u$ \,\\
            \cline{1-3}
            \multirow{2}{*}{\hspace{-4em}$ \ge 1/7$ \hspace{-1em}$ \left.\begin{array}{l}
                \\
                \\
                \end{array}\right\lbrace $}\, $0$ & $0$ & $0$ \, & \\
             \, $0$ & $0$ & $1$ \, & \multirow{6}{*}{\hspace{-1.2em}$\left.\begin{array}{l}
                \\
                \\
                \\
                \\
                \\
                \\
                \end{array}\right\rbrace \ge 1-2\delta$} \\
            \, $0$ & $1$ & $0$ \,  & \\
            \, $0$ & $1$ & $1$ \, & \\
            \, $1$ & $0$ & $0$ \, & \\
            \, $1$ & $0$ & $1$ \, & \\
            \, $1$ & $1$ & $0$ \, & \\
            \, $1$ & $1$ & $1$ \, & \\
            \cline{1-3}
        \end{tabular}
    \end{center}
\caption{Distribution of the probability mass in $\mu_{\{u,v,w\}}$ to assignments of a constraint $P_{\{u,v,w\}}$ that is unsatisfied with probability at most $2\delta$, assuming that $P_{\{u,v,w\}}$ has only positive literals.}
\label{tab:fixing}
\end{table}
\noindent
We now continue to look at a fixed constraint $P_{\{u,v,w\}}$ as above, and consider the following random experiment: draw a random pair of variables $x,y$ together with a random assignment $(\beta_x,\beta_y) \sim \mu_{x,y}$ sampled from their local distribution, and let $\tilde{\mu}=\mu|_{(x,y)\leftarrow (\beta_x,\beta_y)}$ be the pseudodistribution obtained by conditioning on the pair $(x,y)$ taking value $(\beta_x,\beta_y)$. With probability  $1/\binom{n}{2}$ we hit the pair $(v,w)$, i.e., $(x,y)=(v,w)$, and with probability $1/7$ we have $(\beta_x,\beta_y)=(\beta_v,\beta_w)=(0,0)$. This means that the variable $u$ becomes $O(\delta)$-fixed in $\tilde{\mu}$ with probability at least $0.1/\binom{n}{2}$. By \Cref{ins:unsatabs}, we know that for many (i.e., $\Omega(n)$) choices for $u$, there are many (i.e., $\Omega(n^2)$) triples $P_{\{u,v,w\}}$ that can $O(\delta)$-fix $u$, as the one considered in the discussion above. Using the linearity of expectation and Markov's inequality, we can then conclude the following.

\begin{insight}
    \label{ins:constfix}
    For a random pair of variables $x,y$ and a random assignment $(\beta_x,\beta_y) \sim \mu_{x,y}$,  with probability $1/1000$ there are  at least $n/1000$ variables $u$ that become $O(\delta)$-fixed in $\tilde{\mu}=\mu|_{(x,y) \gets (\beta_x,\beta_y)}$.
\end{insight}
\noindent
A formal version of this fact is stated in \Cref{lem:condtofix}.

\paragraph{Advantage: ``completely'' fixing $O(\delta)$-fixed variables incurs little cost.} We recall that conditioning on a random assignment $(\beta_x,\beta_y) \sim \mu_{x,y}$ preserves the objective value in expectation, i.e., $\ex[\val(\tilde{\mu})]=\val(\mu)=\delta$. Hence, \Cref{ins:constfix} seems to suggest that conditioning on just two variables allows making nontrivial progress towards the goal of rounding the pseudodistribution: the entropy of $\Omega(n)$ variables should drastically decrease while preserving the expected objective value. 
However, the conditioning alone is not enough, as its effect decreases when the entropy becomes already small but still nontrivial (e.g., if all but $\sqrt{n}$ variables are $O(\delta)$-fixed, the above argument cannot guarantee any more progress). 
While independent rounding suffices for \maxcsps in this case, its additive loss can still be prohibitive for the minimization objective. 

In order to bypass this barrier, we additionally use {\em thresholding} to round many $O(\delta)$-fixed variables to exactly $0$ or $1$. Formally,  $\hat{\mu}$ is a new pseudodistribution where for each variable $u$ that is $O(\delta)$-fixed in $\tilde{\mu}$, fix $u$ to the bit retaining $1-O(\delta)$ of the probability mass of $\tilde{\mu}_u$. 

We now zoom in on a constraint $P_{\{u,v,w\}}$ with only positive literals where, say, $u$ has been fixed to $1$ in $\hat{\mu}$. Call $\gamma$ the probability that $P_{\{u,v,w\}}$ is unsatisfied by an assignment sampled from $\tilde{\mu}_{\{u,v,w\}}$ for convenience. One can observe that the probability of $P_{\{u,v,w\}}$ being violated by an assignment sampled from $\hat{\mu}_{\{u,v,w\}}$ equals $\hat{\mu}_{\{u,v,w\}}(1,1,1)$. This is because the assignment $(0,0,0)$ has probability mass $0$ in $\hat{\mu}_{\{u,v,w\}}$, since  we fixed $u$ to $1$. Then, with the help of \Cref{tab:rounding}, we can convince ourselves that the probability of $P_{\{u,v,w\}}$ being unsatisfied by an assignment sampled from $\hat{\mu}_{\{u,v,w\}}$ is no larger than $\gamma$ up to an additive $O(\delta)$ error.
\begin{table}[H]
\begin{center}
\hspace{4.5em}
        \begin{tabular}{|@{}c cc@{}| l@{}}
            \cline{1-3}
            \, $u$ & $v$ & $w$ \,\\
            \cline{1-3}
            \, $0$ & $0$ & $0$ \, & \\
             \, $0$ & $0$ & $1$ \, & \multirow{6}{*}{\hspace{-1.2em}$\left.\begin{array}{l}
                \\
                \\
                \\
                \\
                \\
                \\
                \end{array}\right\rbrace = 1-\gamma$} \\
            \, $0$ & $1$ & $0$ \,  & \\
            \multirow{1}{*}{\hspace{-4.5em}$ \le O(\delta)$ \hspace{-1em}$ \left.\begin{array}{l}
                \\
                \end{array}\right\lbrace $}\, $0$ & $1$ & $1$ \, & \\
            \multirow{4}{*}{\hspace{-6.5em}$ \ge 1-O(\delta)$ \hspace{-1em}$ \left.\begin{array}{l}
                \\
                \\
                \\
                \\
                \end{array}\right\lbrace $}\, $1$ & $0$ & $0$ \, & \\
            \, $1$ & $0$ & $1$ \, & \\
            \, $1$ & $1$ & $0$ \, & \\
            \, $1$ & $1$ & $1$ \, & \multirow{1}{*}{\hspace{-1.1em}$\left.\begin{array}{l}
                \\
                \end{array}\right\rbrace \le \gamma$} \\
            \cline{1-3}
        \end{tabular}
    \end{center}
\caption{Distribution of the probability mass in $\tilde{\mu}_{\{u,v,w\}}$ to assignments of a constraint $P_{\{u,v,w\}}$ where $u$ is $O(\delta)$-fixed in $\tilde{\mu}$ and where $\gamma$ is the probability that $P_{\{u,v,w\}}$ is unsatisfied by an assignment sampled from $\tilde{\mu}_{\{u,v,w\}}$, assuming that $P_{\{u,v,w\}}$ has only positive literals.}
\label{tab:rounding}
\end{table}

\noindent
From the above discussion, we conclude that the fractional objective value $\val(\hat{\mu})$ of $\hat{\mu}$ is at most $\val(\tilde{\mu})+O(\delta)$. Recalling that $\ex[\val(\tilde{\mu})]=\val(\mu)=\delta$, \Cref{ins:constfix} can then be refined as follows.

\begin{insight}
    \label{ins:newpd}
    By sampling a random pair of variables, conditioning on a random assignment to them, and performing the thresholding, we can construct a new pseudodistribution $\hat{\mu}$ such that $\ex[\val(\hat{\mu})]=O(\delta)$ and with probability $1/1000$ there are at least $n/1000$ integral variables $u$ in $\hat{\mu}$.
\end{insight}
\noindent
A formal version of this fact can be obtained by combining \Cref{lem:condtofix} and \Cref{lemma:threshround}.

\subsection{Towards an algorithm}
\label{subsec:techoverview_algo}
It would be tempting to employ \Cref{ins:newpd} to get a rounding scheme as the one outlined in \Cref{alg:simple}: sample a pair of variables and a random assignment, construct $\hat{\mu}$ as above, and recurse on the sub-instance induced by the variables that are still not integral (or halt if there is no such variable). The intuition behind this \Cref{alg:simple} is backed up by a trifecta.
 \begin{enumerate}
     \item First, \Cref{ins:newpd} applies also to the sub-instance that we recurse on: it essentially only relies on \Cref{ins:unsatabs}, which also holds on the sub-instance since $\cI$ is complete (using the fact that being complete is a ``hereditary property'', as anticipated).
     \item Second, we expect that $O(\log n)$ levels of recursion will make every variable integral, by \Cref{ins:newpd}.
     \item Third, the expected value of the (integral) pseudodistribution obtained at the end is roughly $O(\delta) \le O(\opt)$, also by \Cref{ins:newpd}.
 \end{enumerate}
 However, more work is needed, since \Cref{ins:newpd} only bounds the cost in expectation at each individual recursion level.

 \begin{algorithm}[!htbp]
    \caption{Conditions on a random pair of variables and thresholds the result.}
    \label{alg:simple}
    \begin{algorithmic}[1]
        \Require a \nae instance $\cI=(V, \cP)$ and a pseudodistribution $\mu$ over $V$.
        \State Let $V_\rU = \{v \in V\ : \ 0<\mu_v(0),\mu_v(1) <1  \}$ be the set of unfixed variables.
        \If{$V_\rU =\emptyset $}
        \State Return the assignment corresponding to $\mu$, i.e., return $\alpha \in \{0,1\}^V$ such that $\mu_v(\alpha_v)=1$.
        \EndIf
        \State Let $\delta = \val(\cI[V_\rU],\mu[V_\rU])$ be the fractional value of $\mu$ on the sub-instance induced by $V_\rU$.
        \State Sample a pair of variables $x,y$ and an assignment $(\beta_x,\beta_y)\sim\mu_{x,y}$.
        \State Let  $\tilde{\mu}$ be the pseudodistribution $\mu$ conditioned on the event that $x,y$ take value $\beta_x,\beta_y$.
        \For{$v \in V_\rU$}
            \If{$v$ is $O(\delta)$-fixed in $\tilde{\mu}$}
                \State Fix $v$ by moving all the mass of $\mu_v$ on the bit $b \in \{0,1\}$ such that $\mu_v(b) \ge 1-\tilde{O}(\delta)$.
            \EndIf
        \EndFor
        \State Recurse on the newly obtained pseudodistribution $\hat{\mu}$.
    \end{algorithmic}
\end{algorithm}

\paragraph{Obstacle: the need for a high probability guarantee.} 
In order to materialize the strategy of \Cref{alg:simple}, we need to control the deviations of $\val(\hat{\mu})$ from $\delta$ at every level of the recursion. More precisely, let $\mu^{(i)}$ be the pseudodistribution at the beginning of the $i$-th recursion level, and let $\delta_i$ the fractional objective value $\val(\mu^{(i)})$ of the pseudodistribution $\mu^{(i)}$. 

For the sake of exposition, here let us ignore the hidden constant in $O(\delta)$ and assume 
$\E[\val(\mu^{(i)})]$ $=\val(\mu^{(i-1)})$
(the ignored error comes from the clauses at least one of whose variables is completely fixed by the 
threshold step, so it can be controlled via the concept of the {\em aggregate value} defined in equation~\eqref{eq:aggregate} of \Cref{sec:roundingalgo}). With this assumption, 
at every level $i$ of the recursion we have
\begin{equation*}
    \ex_{\substack{\text{$(i-1)$-th level}}}\left[\val\left(\mu^{(i)}\right)\right] = \delta_{i-1} = \val\left(\mu^{(i-1)}\right) \, .
\end{equation*}
Applying this identity at every level, we get
\begin{align*}
    \ex_{\substack{\text{ algorithm}}}\left[\val\left(\mu^{(\text{\#levels})}\right)\right] & =
    \ex_{\substack{\text{$1$-st level}}}\left[ \, \dots \, \ex_{\substack{\text{$(\text{\#levels}-2)$-th level}}}\left[
    \ex_{\substack{\text{$(\text{\#levels}-1)$-th level}}}\left[\val\left(\mu^{(\text{\#levels})}\right)\right] \right] \, \dots \, \right]\\
    & =\delta_1 \\
    & \le \opt \, ,
\end{align*}
where the last inequality follows assuming that the pseudodistribution on which we first call the algorithm is optimal. Now, one can intuitively see that the deviation of the last  pseudodistribution $\mu^{(\text{\#levels})}$ from its expected objective value translates to the approximation ratio achieved by the algorithm. Therefore, we are interested in bounding such deviation. To do so, we have to peel off one expectation operator at a time in the above equation by means of Markov's inequality. At every level $i$, the guarantee would then be $\val(\mu^{(i)}) \le \lambda \cdot \val(\mu^{(i-1)})$ with probability $1-1/\lambda$. Since this gives
\begin{equation*}
    \val\left(\mu^{(\text{\#levels})}\right) \le \lambda^{\text{\#levels}} \opt \, ,
\end{equation*}
we need a $\lambda$ such that $\lambda^{\text{\#levels}} \le \tilde{O}(1)$. Recalling that $\text{\#levels} = \Omega(\log n)$, we must hence require $\lambda \le 1+1/\log n $. Unfortunately, Markov's inequality ensures $1+1/\log n$ multiplicative increase with only $p_{\text{val}}=O(1/\log n)$ success probability. If we call $p_{\text{fix}}=1/1000$ the  probability that $\hat{\mu}$ fixes $n/1000$ variables, we realize that
\begin{equation*}
    (1-p_{\text{val}})+(1-p_{\text{fix}}) \ge 1 \, .
\end{equation*}
We are therefore unable to conclude that there exists a choice of $x,y,\beta_x,\beta_y$ for which we both have that $\mu^{(i)}$ fixes $n/1000$ variables and $\mu^{(i)} \le (1+1/\log n) \cdot \mu^{(i-1)}$. For the left-hand side above to be less than one, we need $p_{\text{fix}} \ge 1-O(1/\log n)$.

\paragraph{Advantage: $\polylog(n)$-degree pseudodistributions.} Recalling that we allow our algorithm to run in $n^{\text{polylog}(n)}$ time, we can assume to have a  $\polylog(n)$-degree pseudodistribution at our disposal. Supported by this, the natural thing to do is modifying each level of the recursion as follows: condition on $O(\log \log n)$ random pairs - as opposed to one - and their respective assignments. While this should boost $p_{\text{fix}} \ge 1-O(1/\log n)$, we now have more conditioning steps that could deviate from the expectation. However, suppose we could condition on a random set of $O(\log \log n)$ pairs of variables whose joint local distribution is close (in total variation distance) to the product of their marginal distributions: then, we would both obtain the high probability guarantee and simultaneously have only one deviation from the objective value to control. A standard way to ensure that the joint distribution of $t$ (pairs of) variables is $\epsilon$-close to independent is by additionally conditioning on $\poly(2^t/\epsilon)$ variables~\cite{rt,yz14}. Since we want $t=O(\log \log n)$ and $\epsilon=1/\polylog(n)$, we can afford to do so with our $\polylog(n)$-degree pseudodistribution.
\\~\\
A formal version of this algorithm is presented and analyzed in \Cref{sec:roundingalgo}.
	\section{Notation}\label{sec:prelims}

\paragraph{\csps.} Let $k \ge 2$ be an integer and let $V$ be a set of $n$ variables. A \kcsp is a pair $\cI = (V, \cP)$ where $V$ is a set of variables over the Boolean alphabet $\{0,1\}$, and $\cP$ represents the constraints. Each constraint $P_S\in \cP$ is defined by a $k$-subset of variables $S \in \binom{V}{k}$ and a predicate $P_S:\{0,1\}^S \rightarrow \{0,1\}$, where $1$ means that the constraint is satisfied and $0$ otherwise. We will consider \emph{complete} \kcsps, so that there is a predicate $P_{S} \in \cP$ for every $S \in \binom{V}{k}$. Furthermore, each constraint  $P_{S}$ is unsatisfied for at least one assignment to the variables $S$. Then, for a global assignment $\alpha \in \{0,1\}^V$ to the variables, we let $$\val(\cI,\alpha) = \Pr_{S \sim \binom{V}{k}}[P_S(\alpha_S)=0] \, ,$$ where $S \sim \binom{V}{k}$ indicates that $S$ is sampled uniformly from the elements of $\binom{V}{k}$, and $\alpha_S$ represents the restriction of $\alpha$ to entries in $S$.

In the \emph{decision} version of a $k$-CSP instance $\cI$, the goal is to decide if there is an assignment $\alpha^*$ such that $\val(\cI, \alpha^*)~=~0$.
In the \emph{minimization} version, which we denote by $\textsc{Min-$k$-CSP}$, the goal is to find an assignment $\alpha^*$ which minimizes $\val(\cI, \alpha^*)$, so we define
$$\opt(\cI)= \min_{\alpha \in \{0,1\}^V} \val(\cI,\alpha)\, .$$
Sometimes we will write $\val(\alpha)$ as a shorthand for $\val(\mathcal{I}, \alpha)$ when $ \mathcal{I}$ is clear from the context. For a subset of variables $W \subseteq V$, we denote by $\cI[W]=(W,\cP')$ the sub-instance of $\cI$ induced by the variables $W$.

\paragraph{\nae.} A complete \nae instance is a complete \ncsp{3} $\cI=(V,\cP)$ where each constraint (also called a {\em clause}) 
$P_S \in \cP$ is a ``not all equal'' predicate on three literals, where each literal is either a variable or its negation. For our convenience, we define this formally as follows: for each $S \in \binom{V}{3}$ and each $u \in S$, we have a polarity bijection $P_S^u:\{0,1\} \rightarrow \{0,1\}$ determining the literal pattern for the variable $u$ in $S$, and the constraint $P_S$ is defined as
\begin{equation*}
    P_S(\alpha)=\begin{cases}
        1 \, , \quad & \text{if } \left|\left\{P_S^u(\alpha_u)\right\}_{u \in S}\right|\ge 2 \\
        0 \, , \quad & \text{otherwise}
    \end{cases}
\end{equation*}
for every $\alpha\in \{0,1\}^S$.

\paragraph{Sherali-Adams notation.} For $d \ge k$, we consider the degree-$d$ Sherali-Adams relaxation of $\opt(\cI)$, which can be written as
\begin{equation*}
	\begin{array}{ll@{}ll}
		\text{minimize } & \Pr_{S \sim \binom{V}{k}, \, \alpha \sim \mu_S} [P_S(\alpha)={0}] & \\
		\text{subject to } & \mu \in \cD(V,d)
	\end{array}
\end{equation*}
where $\cD(V,d)$ is the set of pseudodistributions of degree $d$ over the Boolean variables $V$, i.e. every element $\mu$ of $\cD(V,d)$ is indexed by subsets $S \subseteq V$ with $|S| \le d$ where each $\mu_S$ is a distribution over assignments $\alpha \in \{0,1\}^S$ to the variables in $S$, with the property that
\begin{equation*}
\forall S,T \subseteq V, \beta \in \{0,1\}^{S\cap T} \text{ with } |S\cup T| \le d\, , \quad 
    \Pr_{\alpha \sim \mu_S}[\alpha_{S\cap T}=\beta] = \Pr_{\alpha \sim \mu_T}[\alpha_{S\cap T}=\beta] \, .
\end{equation*}
Moreover, for any $\alpha \in \{0,1\}^S$, we let $\mu_S(\alpha) \in [0,1]$ denote that probability of sampling $\alpha$ from $\mu_S$.
We also let for convenience
\begin{equation*}
    \val(\cI,\mu) = \Pr_{S \sim \binom{V}{k}, \, \alpha \sim \mu_S} [P_S(\alpha)={0}]
\end{equation*}
be the fractional objective value of a pseudodistribution $\mu$. Additionally, for $W \subseteq V$, we denote by $\mu[W] \in \cD(W,d)$ the restriction of $\mu$ to variables in $W$. Again, we will use the shorthand $\val(\mu)$ when $\mathcal{I}$ is clear from the context.

\paragraph{Conditioning and fixing notation.} For $\mu \in \cD(V,d)$, $S \subseteq V$ with $|S| \le d$ and $\beta \in \supp(\mu_S)$, we denote by $\mu|_{S\gets \beta} \in \cD(V,d-|S|)$ the pseudodistribution obtained from $\mu$ by conditioning on the event that the variables in $S$ are assigned $\beta$. Formally, $\mu|_{S\gets \beta}$ is the pseudodistribution defined for each $T \subseteq V$ with $|T| \le d-|S|$ and $\alpha \in \{0,1\}^T$ as
\begin{equation*}
    \left(\mu|_{S\gets \beta}\right)_T(\alpha) = \begin{cases}
        \frac{\mu_{S \cup T}(\alpha \oplus \beta)}{\mu_S(\beta)}\, , \quad & \text{if } \alpha_{S \cap T} = \beta_{S \cap T}\\
        0 \, , \quad & \text{otherwise}
    \end{cases}
\end{equation*}
where $\alpha\oplus\beta \in \{0,1\}^{T\cup S}$ is the vector assigning $\alpha$ to $T$ and $\beta_{S\setminus T}$ to $S \setminus T$.
Furthermore, we denote by $\mu^{S\gets \beta} \in \cD(V,d)$ the pseudodistribution obtained by fixing the variables in $S$ to take the assignment $\beta$, which results from moving all the probability mass of $\mu_{\{v_i\}}$ to $\mu_{\{v_i\}}(b_i)$ for each $v_i\in S$ and the corresponding $b_i \in \beta$.
Formally, $\mu^{S\gets \beta}$ is the pseudodistribution defined for each $T \subseteq V$ with $|T| \le d$ and each $\alpha \in \{0,1\}^T$ as

\begin{equation*}
    \left(\mu^{S\gets \beta}\right)_T(\alpha) = \begin{cases}
        \sum_{\substack{\alpha' \in \{0,1\}^T:\\ \alpha'_{T \setminus S}=\alpha_{T \setminus S}}} \, \, {\mu_T(\alpha')} \, , \quad & \text{if } \alpha_{S \cap T} = \beta_{S \cap T}\\
        0 \, , \quad & \text{otherwise}
    \end{cases}
\end{equation*}

\paragraph{Variable vectors and variable sets.}
We will use the convention that for a vector $\vecc \in V^\ell$ written in boldface, the regular capital letter $C =\{ \vecc_i \}_{i \in [\ell]}$ will be the set formed by the union of the value taken by its coordinates. 
	\section{Polylogarithmic approximation for complete Min-NAE-3-SAT}
\label{sec:roundingalgo}

In this section, we present and analyze the rounding scheme  that we apply to the Sherali-Adams solution. Hereafter, we fix an instance $\cI$ of \nae with variable set $V$, where $|V| = n$. In order to describe the rounding algorithm, we should first introduce the necessary notions.

\paragraph{Fixed variables, ruling value, ruling assignment.}  Given a subset $W\subseteq V$, a pseudodistribution $\rho \in \cD(W,d)$, and a threshold
$\xi\in[0,1/2)$, a variable $v$ is called \textit{$(\rho,\xi)$-fixed} if for some bit $b\in\{0,1\}$ one has that the probability $\rho_{\{v\}}(b)$ is at most $\xi$. Moreover, the corresponding bit $1-b$ having probability at least $1-\xi$ is called the \emph{ruling value} of $v$. Also, we 
define 
$F_W(\rho,\xi)$ to be the the set 
of all the $(\rho,\xi)$-fixed variables in $W$. Moreover, we define $\omega^*_W(\nu,\xi) \in \{0,1\}^{F_W(\nu,\xi)}$ to be the corresponding \textit{ruling assignment}, where the entry for $v \in F_W(\nu,\xi)$ equals the ruling value of $v$ in~$\rho$.

\paragraph{LP value of constraint classes.} Given a set $V_\rU \subseteq V$ of variables (which can be thought of as the set of unfixed variables), we classify the constraints of $\cI=(V,\cP)$ into four groups: for $i \in \{0,1,2,3\}$, we let $\cI_i\{V_\rU\}=(V,\cP_i\{V_\rU\})$ be the sub-instance with constraints $\cP_i\{V_\rU\}=\{P_S \in \cP: |S\cap V_\rU|=i\}$, i.e. $\cI_i\{V_\rU\}$ can be thought of as the sub-instance consisting of constraints with exactly $i$ unfixed variables. Then, given a pseudodistribution $\rho \in \mathcal{D}(V,d)$, we define for each $i \in \{0,1,2,3\}$ the \textit{LP value of the constraints in $\cI_i\{V_\rU\}$} as
\begin{equation*}
    \lp^{V_\rU}_i(\rho) = \sum_{P_S \in \cP_i\{V_\rU\}} \Pr_{\alpha \sim \rho_S}[P_S(\alpha)=0] \, ,
\end{equation*}
\noindent
so $\lp^{V_\rU}_i(\rho)$ can be thought of as the contribution of the constraints with exactly $i$ unfixed variables to the Sherali-Adams objective.

\paragraph{Aggregate LP value.} For analyzing our algorithm, we treat the quantities $\lp^{V_\rU}_0(\rho)$, $\lp^{V_\rU}_1(\rho)$, $\lp^{V_\rU}_2(\rho)$ and $\lp^{V_\rU}_3(\rho)$ separately. However, for the sake of certain probability bounds, it turns out to be useful to combine these into a single weighted sum:
given a real parameter $\tau>0$, a set $V_\rU \subseteq V$, and a pseudodistribution $\rho \in \mathcal{D}(V,d)$, we define the \emph{aggregate value} of $\rho$ as the weighted sum

\begin{equation}
A^{V_\rU}_\tau(\rho) = \tau \cdot \log^3 n \cdot \lp^{V_\rU}_3(\rho) + 
\log^2 n \cdot \lp_2^{V_\rU}(\rho) + 
\log n \cdot \lp^{V_\rU}_1(\rho) + 
\lp^{V_\rU}_0(\rho) \, . \label{eq:aggregate}
\end{equation}

\noindent
We might drop the subscript $\tau$ when it is clear from context (and this should not create any confusion as we will use one fixed value of $\tau$ for the whole algorithm and analysis).

\paragraph{Thresholds with bounded increase.} As per the definition above, the notion of fixed variables depends on a threshold $\xi$. In our algorithm and analysis, we will want to use a threshold that ensures a certain bound on the additive increase on the aggregate value. More precisely, given $\tau>0$, $V_\rU \subseteq V$, $\rho \in \mathcal{D}(V,d)$, and a parameter $\delta \in [0,1]$, we define the \textit{set of thresholds with bounded increase} to be the set of all $\theta \in [0,1/2)$ such that
\begin{equation*}
    A^{V_{\rU}\setminus F_{V_\rU}(\rho,\theta)}(\rho^{F_{V_{\rU}}(\rho,\theta) \gets \omega^*_{V_{\rU}}(\rho,\theta)})- A^{V_{\rU}}(\rho)\leq\frac{6}{\log n}A^{{V_{\rU}}}(\rho)+ 12\tau \delta\log^2 n \binom{|V_{\rU}|}{3} \, ,
\end{equation*}
and we denote this set as $\Theta_{\tau}^{V_\rU}(\rho,\delta)$. As for $A^{V_\rU}_\tau(\rho)$, we drop the subscript $\tau$ when it is clear from context.

\paragraph{$2$-SAT instances.} Given $V_\rU \subseteq V$ and a partial assignment $\alpha \in \{0,1/2,1\}^V$ such that $\alpha_v \in \{0,1\}$ for all $v \in V\setminus V_\rU$, the \mintwo instance induced by $V_\rU$ and $\alpha$ is denoted as $\cI_2\langle V_\rU,\alpha \rangle $ and is defined as follows: we let $\cI_2\langle V_\rU,\alpha \rangle = (V_\rU,\cP_2\langle V_\rU, \alpha\rangle)$, where $\cP_2\langle V_\rU, \alpha\rangle$ is a multiset of \twosat constraints that associates to each $P_S \in \cP_1\{V_\rU\} \cup \cP_2\{V_\rU\}$ exactly one constraint $P'_S:\{0,1\}^{\{x,y\}}\rightarrow\{0,1\}$ defined as $P'_S(\alpha')=P_S(\alpha'\oplus\alpha_{S\setminus V_\rU})$ for all $\alpha' \in \{0,1\}^{S \cap V_\rU}$, where $\alpha'\oplus\alpha_{S\setminus V_\rU} \in \{0,1\}^S$ is the vector assigning $\alpha'$ to $S \cap V_\rU$ and $\alpha_{S \setminus V_\rU}$ to $S \setminus V_\rU$. We remark that since all $P_S \in \cP$ are $\nae$, every $P'_S \in \cP_2\langle V_\rU, \alpha\rangle$ is a \twosat or \nsat{1} clause.

\subsection{Algorithm outline}
\Cref{alg:roundingalgo} gives our rounding scheme. Throughout its execution, the input instance $\cI=(V,\cP)$ remains unchanged, as do the parameters $r,t,\epsilon,\tau$, which only depend on $n=|V|$ and on the fixed universal constant $K = 10^{20}$.

\begin{algorithm}[ht]
    \caption{$\roundalgo(\cI,d,\mu,\alpha,\ell)$}
    \label{alg:roundingalgo}
    \begin{algorithmic}[1]
        \Params $K = 10^{20},\, r = (\log n)^{100K}, \, t = K \log \log n, \,\epsilon = {1}/({10\log n}), \, \tau=K \log^{2} n$.
        \Require $\cI=(V, \cP), \, d \in \N,  \,  \mu \in \cD(V,d) , \, \alpha \in \{0,1/2,1\}^V, \, \ell \in \N$.
        \State Let $V_\rF = \{v \in V\ : \alpha_v \in \{0,1\} \}$ and $V_\rU = V\setminus V_\rF$.
        \If{$|V_\rU|< K^2(\log n)^{100}$} \Comment{few unfixed variables left}
            \State \Return $\alpha^* \in \arg \min_{\alpha' \in \{0,1\}^V:\, \alpha'_{V_\rF}=\alpha_{V_\rF}} \val(\cI,\alpha')$. \label{line:bruteforce} \Comment{brute force}
        \EndIf
        \State Let $\delta = \val(\cI[V_\rU],\mu[V_\rU])$.\label{step:delta}
        \If{$\delta > 1/(10\tau) $} \Comment{constraints induced by unfixed variables are not important}
        \State Let $ \alpha' = \textsc{kprt}{}(\cI_2\langle V_\rU,\alpha \rangle,\mu[V_\rU]) \in \{0,1\}^{V_\rU}$. \Comment{use the  \mintwo rounding of \cite{klein1997approximation}}
        \State Let $\alpha^* \in \{0,1\}^V$ be defined as $\alpha^*_{V_\rF}=\alpha_{V_\rF}$ and $\alpha^*_{V_\rU}=\alpha'$.
        \State \Return $\alpha^*$ \label{line:2satinstance}
        \EndIf
        \State Let $\tilde{\mu} = \mu$. \Comment{initialization of conditioned pseudodistribution}
        \For{$r' =0, \dots, r$}\label{line:conditioningst}
        \For{ $\vecc \in V_\rU^{r' + 2t}$,  $\gamma \in  \supp\left(\mu_C\right)$}
            \If{$|F_{V_\rU}(\mu|_{C \gets \gamma }, \, \tau \delta)| \ge \frac{|V_\rU|}{100}$ and $A^{V_\rU}(\mu|_{C \gets \gamma}) \leq (1 + \epsilon) A^{V_\rU}(\mu)$}
            \State Update $\tilde{\mu} = \mu_{C \gets \gamma}$. \Comment{conditioning step}
            \EndIf
        \EndFor
        \EndFor\label{line:conditioningend}
        \State Let $\theta \in \Theta^{V_\rU}(\tilde{\mu},\delta)$ if $ \Theta^{V_\rU}(\tilde{\mu},\delta) \neq \emptyset$, otherwise let $\theta=0$. \label{line:thresholdingst}
        \State Let $F = F_{V_\rU}(\tilde{\mu},\theta) $ and  $\omega = \omega_{V_\rU}^*(\tilde{\mu},\theta)$. 
        \State Let $\hat{\mu} = \tilde{\mu}^{F \gets \omega}$.  \Comment{thresholding step} \label{step:threshround}        
        \State For each $v \in F$, update $\alpha_v = \omega_v$. \label{line:thresholdingend}
        \State \Return $\roundalgo(\cI,d-r-2t,\hat{\mu},{\alpha},\ell+1)$        
    \end{algorithmic}
\end{algorithm}

The algorithm maintains a Sherali-Adams solution $\mu$ and a partial assignment $\alpha$. At any intermediate point, there are some variables $V_\rF$ which are completely fixed: these are variables $v$ for which $\alpha_v$ has already been determined, and $\alpha_v \in \{0,1\}$. The remaining variables $V_\rU$ are the unfixed variables: these are the variables $v$ for which $\alpha_v = \frac{1}{2}$, i.e. their assignment is not yet decided. The algorithm is called as \smash{$\textsc{round-pd}(\cI, d^*, \mu^*, \alpha^{(0)},0)$} where the Sherali-Adams solution $\mu^*$ is obtained by solving a $d^*$-degree Sherali-Adams LP for \minnaethree where $d^* = \polylog(n)$ will be determined in the analysis below.

On a high level, our algorithm proceeds in stages. At each stage, if the number of unfixed variables is small (specifically $|V_{\rU}| < K^2 (\log n)^{100}$), we simply solve the instance by brute force (line~\ref{line:bruteforce}). Otherwise, if the value of the instance induced on $V_{\rU}$ is small, $\val(\cI[V_{\rU}], \mu[V_{\rU}]) > {1}/({10\tau})$, we show that the instance is essentially equivalent to $2$-SAT, and solve it using a standard LP-rounding for $2$-SAT (line~\ref{line:2satinstance}, see~\Cref{lem:highlyunsatok}). The interesting case is when neither of these two cases happens: we proceed via a conditioning step and a thresholding step, followed by a recursive call. In the conditioning step (lines \ref{line:conditioningst}-\ref{line:conditioningend}), the algorithm identifies a small set of variables $C$ and an assignment to these variables $\gamma$, and conditions on the event that $C$ gets assigned $\gamma$, to obtain a new pseudodistribution $\tilde{\mu}$. Ideally, this pseudodistribution $\tilde{\mu}$ is such that  many of the variables in $V_\rU$ are $(\tilde{\mu},\tau \delta)$-fixed. The thresholding step (lines~\ref{line:thresholdingst}-\ref{line:thresholdingend}) then rounds these variables to their ruling assignment, resulting in a new pseudodistribution $\hat{\mu}$. After updating $\alpha$ for each newly rounded variable to its already integral value in $\hat{\mu}$, the algorithm continues by making a recursive call.

\subsection{Analysis}
In each stage, after the thresholding step is completed, we will show that the number of unfixed variables (the variables with $\alpha_v = \frac{1}{2}$) decreases by a constant factor. Hence the number of stages is at most $O(\log n)$. Thus if we can show that the conditioning and thresholding steps do not increase the value of the Sherali-Adams solution $\mu$ by more than a $1 + O(1/\log n)$ factor, we would obtain an integral assignment to the variables with a constant-approximation ratio at the end of $O(\log |V|)$ stages. While it is unclear if such a scheme is possible, we track the growth of the aggregate value $A_\tau^{V_\rU}(\mu) = \Theta(\polylog(n)\val(\mu))$, and show that $A_\tau^{V_\rU}(\mu)$ does not increase by more than a $1+O({1}/{\log n})$ factor after each stage. Hence, $A_\tau^{V_\rU}(\mu)$ always remains a constant approximation to $A_\tau^{V}(\mu^*$), where $\mu^*$ is the initial Sherali-Adams solution. By the definition of the aggregate value, a constant-approximation to $A_\tau^{V_\rU}(\mu)$ must imply a $\polylog(n)$-approximation, giving us the desired result. The rest of this section is devoted to formalizing and proving the above statements.

We begin by analyzing the conditioning step. Informally, we show that if we condition on a small random set of variables and a suitable assignment sampled from their local distribution, then many variables will be fixed with good probability in the obtained conditional pseudodistribution. More precisely, we have the following lemma, whose proof is deferred to \Cref{sec:condtofix}. 
\begin{restatable}[Conditioning-to-Fixing]{lemma}{condtofix}
\label{lem:condtofix}
    Let $d,r,t$ be positive integers,  let $W \subseteq V$, let $\rho \in \cD(W,d)$, and let $\tau \ge 1$. If $\val(\cI[W],\rho) \le 1/(3\tau)$, $|W| \ge 3$, $d  \ge r+2t+3$, $r \ge 10^3 3^{2t}$, $t \ge 10^6$, then there exists $r' \in \{0,\dots,r\}$ such that
    \begin{equation*}
    \Pr_{
        \substack{
            \cC \sim W^{r'+2t}, \, \gamma \sim \rho_C
        }
    }
    \left[
        \left|F_W(\rho|_{C \gets \gamma }, \, \tau \cdot \val(\cI[W],\rho))\right| \ge \frac{|W|}{100}
    \right]
    \ge
    1-\left(\frac{10^7}{\sqrt{\tau}}+5 \frac{3^{2t/3}}{ r^{1/3}}+\frac{5t^2}{|W|}+3 \eul^{-t/10^6}\right)
    \, .
    \end{equation*}
\end{restatable}

\noindent
Next, we analyze the change in the aggregate value from the thresholding step. Informally, we bound the additive increase of the aggregate value after line~\ref{line:thresholdingend} compared to the aggregate value before line~\ref{line:thresholdingst}. More precisely, we have the following lemma, whose proof is deferred to \Cref{sec:threshround}.

\begin{restatable}[Thresholding]{lemma}{threshround}
\label{lemma:threshround}
    Let $d$ be  a positive integer, let $\rho \in \cD(V,d)$, let $V_\rU \subseteq V$ such that for all $v \in V\setminus V_\rU$ there is $b \in \{0,1\}$ satisfying $\rho_{\{v\}}(b)=1$, and let $\tau \ge 1$ and $\delta \in [0,\frac{1}{10\tau}]$. Then, there exists  $\theta\in[\tau\delta, 2\tau\delta]$ such that $\theta 
    \in \Theta^{V_\rU}_\tau(\rho,\delta)$. Moreover, such a value can be found in time $\poly(|V_U|)$.
\end{restatable}

\noindent
The idea is to combine the two lemmas above to ensure that when we reach the base of the recursion we have a pseudodistribution $\mu$ whose aggregate value is within a constant factor of the original one.

\begin{lemma}
\label{lem:returnpoints}
    Let  $c = 200K^2$, let $\mu^* \in \cD(V,K^2\log^c n)$, let $\alpha^{(0)} \in \{0,1/2,1\}^V$ with $\alpha_v = 1/2$ for all $v \in V$. Then, running \Cref{alg:roundingalgo} as $\roundalgo(K^2\log^{c} n,\cI,\mu^*,\alpha^{(0)},0)$ must reach line~\ref{line:bruteforce} or~\ref{line:2satinstance} with integers $d\ge 3, \ell \le 100\log n$, a pseudodistribution $\mu \in \cD(V,d)$, and a set $V_{\rU} \subseteq V$ such that
    \begin{equation*}
         A^{{V_\rU}}(\mu) \leq \eul^{5000} A^{V}(\mu^*) = O(1)A^{V}(\mu^*) \, .
    \end{equation*}
\end{lemma}
\begin{proof}
    Notice that before we reach lines~\ref{line:bruteforce} or~\ref{line:2satinstance}, every recursive call to~\Cref{alg:roundingalgo} involves one conditioning step and one thresholding step. Let us consider in particular a call of \Cref{alg:roundingalgo} that starts from a pseudodistribution $\mu \in \cD(V,d)$, and a set of unfixed variables $V_\rU \subseteq V$ such that neither line~\ref{line:bruteforce} nor~\ref{line:2satinstance} is reached. Then, the conditioning step produces a pseudodistribution $\tilde{\mu}$. This step does not change the assignment $\alpha$, or the set of unfixed variables $V_\rU$. Next, the thresholding step rounds some variables in $\tilde{\mu}$ and produces a pseudodistribution $\hat{\mu}$, and updates the assignment $\alpha$. Let us call this updated assignment $\alpha'$ and let $V_\rU' = \{v \in V \mid \alpha'_v = \frac{1}{2}\}$ be the set of unfixed variables with respect to $\alpha'$. For the sake of convenience, let us set define $A(\mu) := A^{V_{\rU}}(\mu)$, $A(\tilde{\mu}) := A^{V_{\rU}}(\tilde{\mu})$, and $A(\hat{\mu}) := A^{V_{\rU}'}(\hat{\mu})$. In words, $A(\mu)$ and $A(\tilde{\mu})$ are defined with respect to the set $V_\rU$, whereas $A(\hat{\mu})$ is defined with respect to the set $V_\rU'$ obtained after thresholding. We will show two things: first, we argue that $|V_\rU'| \le 99/100 \cdot |V_\rU|$, and second we argue that the increase in the aggregate value $A(\hat{\mu}) - A(\mu)$ is at most $50 A(\mu)/\log n$.
    
    Recall our setting of parameters: $K = 10^{20}$, $r = (\log n)^{100K}$, $t = K \log \log n$, $\epsilon = {1}/{10\log n}$, $\tau =  K\log^{2} n$, $c = 200K^2$.
    Since lines~\ref{line:bruteforce} and~\ref{line:2satinstance} are not reached, we can assume that $|V_\rU| \geq c_1 \geq K^2\log^{100} n$ and that $\val(\cI[V_\rU],\mu[V_\rU]) \leq \frac{1}{10\tau}$. By the guarantee of~\Cref{lem:condtofix} with $W = V_\rU$ and $\rho=\mu[V_\rU]$, there exists $r' \in \{0,\dots,r\}$ such that

    \begin{align*}
    \Pr_{
        \substack{
            \cC \sim {V_\rU}^{r'+2t}, \\ \gamma \sim \mu_C
        }
    }
    \left[
        \left|F_{V_\rU}(\mu|_{C \gets \gamma }, \, \tau \cdot \val(\cI[V_\rU],\mu[V_\rU]))\right| \ge \frac{|V_\rU|}{100}
    \right]
    &\ge
    1-\left(\frac{10^7}{\sqrt{\tau}}+5 \frac{3^{2t/3}}{ r^{1/3}}+\frac{5t^2}{|V_\rU|}+3 \eul^{-t/10^6}\right)\\
    &\geq 1 - \frac{1}{100\log n} 
    \, .
    \end{align*}
\noindent
We also observe that for all $i \in \{0,1,2,3\}$ we have
\begin{equation*}
    \E_{        \substack{
            \cC \sim V_\rU^{r'+2t}, \, \gamma \sim \mu_C
        }}[\lp^{V_{\rU}}_i(\mu|_{C \gets \gamma})] = \lp^{V_{\rU}}_i \, ,
\end{equation*}
which follows from the property of conditional expectation\footnote{Note that while the conditioning ``fixes" the variables we condition on, we still do not add these variables to $V_\rU$ - this happens only after the thresholding step. Hence, the sets $V_\rU$ and $V_\rF$ remain unchanged, and therefore a constraint which contributes to $\lp_i^{V_{\rU}}(\mu)$ continues to contribute to $\lp_i^{V_{\rU}}(\mu|_{C \gets \gamma})$ for all $i \in \{0,1,2,3\}$.}. This in turn means that
\begin{equation*}
    \E_{        \substack{
            \cC \sim V_\rU^{r'+2t}, \, \gamma \sim \mu_C
        }}[A^{V_{\rU}}(\mu|_{C \gets \gamma})] = A^{V_{\rU}}(\mu) = A(\mu) \, .
\end{equation*}
By Markov's inequality, it follows that

\begin{align*}
\Pr_{        \substack{
            \cC \sim W^{r'+2t}, \, \gamma \sim \mu_C
        }}[A^{V_{\rU}}(\mu|_{C \gets \gamma}) \geq A(\mu)(1 + \epsilon)] \leq 1 - \epsilon/2 = 1 - \frac{1}{20\log n}.
\end{align*}

\noindent
It follows that there exists a choice of $\cC \in W^{r'+2t}$ and $\gamma \in \supp(\mu_C)$ such that we simultaneously have

$$
        \left|F_{V_\rU}(\mu|_{C \gets \gamma }, \, \tau \cdot \val(\cI[V_\rU],\mu))\right| \ge \frac{|V_\rU|}{100}
     \quad \mbox{ and } \quad  A^{V_\rU}(\mu|_{C \gets \gamma}) \leq A(\mu)(1 + \epsilon)  \, .$$

\noindent
In particular, we must have
$$
        \left|F_{V_\rU}(\tilde{\mu}, \, \tau \cdot \val(\cI[V_\rU],\mu))\right| \ge \frac{|V_\rU|}{100}
     \quad \mbox{ and } \quad  A(\tilde{\mu}) \leq A(\mu)(1 + \epsilon)  \, .$$

\noindent
    Next, we analyze the rounding step. Note that the pseudo-distribution $\hat{\mu}$ is obtained by rounding $\tilde{\mu}$. By using \Cref{lemma:threshround} with $\rho = \tilde{\mu}$, it follows that
    
    $$ A(\hat{\mu}) - A(\tilde{\mu}) \leq \frac{6}{\log n} A(\tilde{\mu}) + 12 \tau \delta \binom{|V_\rU|}{3} \log^2 n$$
\noindent
    Note that $\tau \delta \binom{|V_\rU|}{3} \log^2 n = \tau \cdot \lp_3^{V_{\rU}}(\mu) \log^2 n \leq \frac{A(\mu)}{\log n}$, since $\delta \binom{|V_U|}{3} = \lp_3^{V_U}(\mu)$. This gives
    $$A(\hat{\mu}) - A(\tilde{\mu}) \leq 20\frac{ A(\tilde{\mu}) + A(\mu)}{\log n} \, .$$
    Using the fact that $A(\tilde{\mu}) \leq (1 + \epsilon) A(\mu)$, it follows that
    $$A(\hat{\mu}) - A(\tilde{\mu}) \leq 45 \frac{A(\mu)}{\log n} \, . $$
    Finally, combining the conditioning and rounding steps, we obtain
    \begin{align*}
    A(\hat{\mu}) - A(\mu) &= A(\hat{\mu}) - A(\tilde{\mu}) +A(\tilde{\mu}) - A(\mu) \\
    &\leq 45A(\mu)/\log n + \epsilon A(\mu) \\
    &\leq 50A(\mu)/\log n,
    \end{align*}
    where the last inequality follows since $\epsilon ={1}/{10\log n}$.

    Note that after each step of conditioning and rounding described above, we fix a $\frac{1}{100}$ fraction of the unfixed variables $V_\rU$. It follows that after at most $100 \log n$ such stages of conditioning and thresholding, the distribution $\mu$ and the set of unfixed variables $V_{\rU}$ satisfy $A^{V_{\rU}}(\mu) \leq (1 + 50/\log n)^{100 \log n} A^{V}(\mu^*) \leq \eul^{5000} A^{V}(\mu^*)$.
\end{proof}
    
\noindent
The following two lemmas, whose proofs are deferred to \Cref{sec:edgecases1,sec:edgecases2} respectively, prove the guarantees for the two base cases in \Cref{alg:roundingalgo}. 
When the number of unfixed variables becomes sufficiently small (line \ref{line:bruteforce}) we can compute the optimal solution using brute force (\Cref{lem:bruteforceok}).
When $\delta = \lp_3(\mu)/\binom{|V_\rU|}{3}$ is large (line \ref{line:2satinstance}), we can simply ignore clauses whose all three variables are unfixed, which reduces the problem to \textsc{Min-2-SAT} (\Cref{lem:highlyunsatok}) that has a well-known $\polylog(n)$-approximation. 

\begin{restatable}[Brute force base case]{lemma}{bruteforce}
    \label{lem:bruteforceok}
Let  $c = 200K^2$, let $\mu^* \in \cD(V,K^2\log^c n)$, let $\alpha^{(0)} \in \{0,1/2,1\}^V$ with $\alpha_v = 1/2$ for all $v \in V$. Also let $d$, $\mu \in \cD(V,d), \alpha \in \{0, 1/2, 1\}^V$ and $V_\rU = \{v : \alpha_v =1/2\}$ be the degree, the pseudodistribution, the partial assignment and the corresponding set of unfixed variables when~\Cref{alg:roundingalgo}, run as $\roundalgo(K^2\log^{c} n,\cI,\mu^*,\alpha^{(0)})$, reaches line~\ref{line:bruteforce}. Then, the output $\alpha^*$ of line \eqref{line:bruteforce} satisfies $\val(\cI,\alpha^*) \le \val(\cI,\mu)$.
\end{restatable}

\begin{restatable}[\mintwo base case]{lemma}{twosatcase}
    \label{lem:highlyunsatok}
    Let  $c = 200K^2$, let $\mu^* \in \cD(V,K^2\log^c n)$, let $\alpha^{(0)} \in \{0,1/2,1\}^V$ with $\alpha_v = 1/2$ for all $v \in V$. Also let $d$, $\mu \in \cD(V,d), \alpha \in \{0, 1/2, 1\}^V$ and $V_\rU = \{v : \alpha_v = 1/2\}$ be the degree, the pseudodistribution, the partial assignment and the corresponding set of unfixed variables when~\Cref{alg:roundingalgo}, run as $\roundalgo(K^2\log^{c} n,\cI,\mu^*,\alpha^{(0)})$, reaches line~\ref{line:2satinstance}. Then, the output $\alpha^*$ of line~\eqref{line:2satinstance} satisfies $\val(\cI,\alpha^*) \leq O(\log n)A^{V_{\rU}}(\mu)/\binom{|V_\rU|}{3}$.
\end{restatable}

\noindent
We finish the proof of our main result. 

\begin{proof}[Proof of \Cref{thm:nae}]
Let $\mu^* \in \cD(V, K^2 \log^c n)$ be obtained by solving a Sherali-Adams relaxation.
 Run~\Cref{alg:roundingalgo} with $c = 200K^2$ as \textsc{round-pd}{}$(K^2\log^c n, \cI, \mu^*, \alpha^{(0)})$, where $\alpha^{(0)} \in \{0,1/2,1\}^V$ with $\alpha_v = 1/2$ for all $v \in V$. \Cref{lem:returnpoints} guarantees that one of lines~\eqref{line:bruteforce} or~\eqref{line:2satinstance} is reached with a pseudodistribution $\mu$ and a set of unfixed variables $V_{{\rU}}$ such that $A^{V_\rU}(\mu) \leq O(1) A^{V}(\mu^*)$. Now there are two cases depending on which line is reached when the algorithm terminates.

 \begin{itemize}
 \item \textbf{Line \eqref{line:bruteforce} is reached.} In this case, by~\Cref{lem:bruteforceok}, we obtain an assignment $\alpha^*$ with $\val(\cI, \alpha^*) \leq \val(\cI, \mu)$. Note that 

$$\val (\cI, \mu) \leq \frac{1}{{\binom{|V|}{3}}}\sum_{i = 0}^3 \lp^{V_{\rU}}_i(\mu)
\leq \frac{1}{{\binom{|V|}{3}}} A^{V_{\rU}}(\mu) \, .$$
This quantity is at most
$$ \frac{1}{{\binom{|V|}{3}}} O(1) A^{V}(\mu^*) \leq \frac{1}{{\binom{|V|}{3}}} O(\tau \log^3 n) \sum_{i = 0}^{3} \lp_i^{V}(\mu^*) \, ,$$
which is equal to
$$\frac{1}{{\binom{|V|}{3}}} O(\tau \log^3 n) \lp^{V}_3(\mu^*) \leq O(\tau \log^3 n) \val(\cI, \mu^*)$$
as desired.

 \item \textbf{Line \eqref{line:2satinstance} is reached.} In this case, by~\Cref{lem:highlyunsatok}, we obtain an assignment $\alpha^*$ with $\val(\cI,\alpha) \leq O(\log n)A^{V_{\rU}}(\mu)/\binom{|V_\rU|}{3}$. Again, we have 

$$\val (\cI, \alpha^*) \leq \frac{1}{{\binom{|V|}{3}}} O(\log n)A^{V_{{\rU}}}(\mu) \leq \frac{1}{{\binom{|V|}{3}}}O(\log n) A^V(\mu^*) \, ,$$
which is at most
$$\frac{1}{{\binom{|V|}{3}}} O(\tau \log^4 n) \sum_{i = 0}^{3} \lp^V(\mu^*) = \frac{1}{{\binom{|V|}{3}}} O(\tau \log^4 n) \lp^V_3(\mu^*) \, .$$
By definition, the above is $O(\tau \log^4 n) \val(\cI, \mu^*)$. Finally, we recall that $\tau = K\log^2 n = O(\log^2 n)$, so that $\val (\cI, \beta^*)\leq O(\log^6 n) \val(\cI, \mu^*).$  This concludes the proof of the correctness.
 \end{itemize}

\noindent
 For the running time, note that in a given stage of the algorithm involving one conditioning and one thresholding step, the bottleneck is guessing the set of $r' + 2t$ variables and their assignments to condition on. Since $r' + 2t = O(\log^{100K} n)$, this step takes quasi-polynomial time. There are at most $100\log n$ stages by \Cref{lem:returnpoints}, thus the total time remains quasi-polynomial. Finally, we can obtain a degree-$d$ Sherali-Adams solution in $n^{O(d)}$ time. Recall that we solve a $K^2\log^c n$-round Sherali-Adams solution to obtain $\mu^*$, where $c = 200K^2$. Thus, the running time for this step is quasi-polynomial as well.  It follows that the overall running time is quasi-polynomial.
 \end{proof}
	\section{Proof of the Conditioning-to-Fixing Lemma}
\label{sec:condtofix}

The focus of this section is proving \Cref{lem:condtofix}. The lemma states that conditioning on a small random set of variables, together with an assignment sampled from their local distribution, yields a new pseudodistribution where an $\Omega(1)$ fraction of variables are almost fixed. The precise claim is restated here for convenience of the reader. To parse the statement, it  helps to think of the parameters as being $d,r,\tau=\polylog(n)$, $t=O(\log \log n)$, and $\val(\cI[W],\mu) = \tilde{O}(\opt) \ll 1/\polylog(n)$.
\condtofix*

\noindent
To show this, we crucially use a lemma of Yoshida and Zhou~\cite{yz14}, stating that the average correlations in the local distribution of size-$t$ sets is $\poly(2^t/r)$, in expectation over the choice of a random set of $r$ variables and an assignment sampled from their local distribution. Again, it is helpful to think of $r=\polylog(n),t=O(\log \log n)$. Formally, we employ the following fact, where for any two finitely supported distributions $\nu_1,\nu_2$, we denote by $\dkl(\nu_1 \parallel \nu_2)$ the KL divergence of $\nu_1$ from $\nu_2$. 
\begin{lemma}[Lemma 6 in~\cite{yz14}]
\label{lem:yoshidazhou}
    Let $d,r,t$ be positive integers such that $d \ge r+2t$, let $W \subseteq V$, and let $\rho \in \cD(W,d)$. Then, there exists $r' \in \{0,1,\dots,r\}$ such that
    \begin{equation*}
        \ex_{\cS \sim W^{r'}, \, \sigma \sim  \rho_{S}}\left[\ex_{\cT \sim W^{2t}}\left[\dkl\left((\rho|_{S \gets \sigma})_T \, \middle\| \, \otimes_{v \in T}(\rho|_{S \gets \sigma})_{\{v\}}\right)\right]\right] \le \frac{3^{2t}}{r} \, .
    \end{equation*}
\end{lemma}
\noindent
The above result in particular implies that with good probability the local distribution of a random size-$t$ set is $\poly(2^t/r)$-close (in total variation distance) to independent in the pseudodistribution obtained by conditioning a random set of $r$ variables and a random assignment. More precisely, \cref{lem:yoshidazhou} readily gives the following fact, where for any two distributions $\nu_1,\nu_2$ supported over a finite set $\Omega$ we denote by $\|\nu_1-\nu_2 \|_1=\sum_{\omega \in \Omega}|\nu_1(\omega)-\nu_2(\omega)|$ the total variation distance between $\nu_1$ and $\nu_2$.
\begin{lemma}
\label{lem:tvd}
    Let $d,r,t$ be positive integers such that $d \ge r+2t$, let $W \subseteq V$, and let $\rho \in \cD(W,d)$. Then, there exist $r' \in \{0,1,\dots,r\}$ such that
    \begin{equation*}
        \Pr_{\substack{\cS \sim W^{r'}, \, \sigma \sim  \rho_{S}, \\ \cT \sim W^{2t}}} \left[\left\|(\rho|_{S \gets \sigma})_T - \otimes_{v \in T}(\rho|_{S \gets \sigma})_{\{v\}}\right\|_1 > 2\cdot \frac{3^{2t/3}}{ r^{1/3}}\right] \le \frac{3^{2t/3}}{ r^{1/3}} \, .
    \end{equation*}
\end{lemma}
\begin{proof}
As guaranteed by \Cref{lem:yoshidazhou}, let $r' \in \{0,1,\dots,r\}$ such that
    \begin{equation*}
        \ex_{\substack{\cS \sim W^{r'}, \, \sigma \sim  \rho_{S}, \\ \cT \sim W^{2t}}}\left[\dkl\left((\rho|_{S \gets \sigma})_T \, \middle\| \, \otimes_{v \in T}(\rho|_{S \gets \sigma})_{\{v\}}\right)\right] \le \frac{3^{2t}}{r} \, .
    \end{equation*}
    We recall that for any two finitely supported distributions $\nu_1,\nu_2$ one has $\|\nu_1-\nu_2\|_1\le \sqrt{2\dkl(\nu_1\parallel \nu_2)}$ (see for example Lemma 2.1 in~\cite{yz14}). Hence, it suffices to bound
    \begin{equation*}
        \Pr_{\substack{\cS \sim W^{r'}, \, \sigma \sim  \rho_{S}, \\ \cT \sim W^{2t}}} \left[\dkl\left((\rho|_{S \gets \sigma})_T \, \middle\| \, \otimes_{v \in T}(\rho|_{S \gets \sigma})_{\{v\}}\right) > 2\cdot \frac{3^{4t/3}}{ r^{2/3}}\right] \, .
    \end{equation*}
    Using Markov's inequality, we can upper bound the above probability by
    \begin{equation*}
        \frac{\ex_{\substack{\cS \sim W^{r'}, \, \sigma \sim  \rho_{S}, \\ \cT \sim W^{2t}}}\left[\dkl\left((\rho|_{S \gets \sigma})_T \, \middle\| \, \otimes_{v \in T}(\rho|_{S \gets \sigma})_{\{v\}}\right)\right]}{2\cdot \frac{3^{4t/3}}{ r^{2/3}}} \le \frac{\frac{3^{2t}}{r}}{2\cdot \frac{3^{4t/3}}{ r^{2/3}}} \le \frac{3^{2t/3}}{ r^{1/3}} \, ,
    \end{equation*}
    which concludes the proof.
\end{proof}

\noindent
Equipped with this tool, we can now proceed to prove several auxiliary lemmas, before concluding a proof of \Cref{lem:condtofix}. To do so, it will be convenient to define the notion of a \textit{fixable} variable. Informally, this is a variable $u$ such that $\Omega(n^2)$ of its constraints have the property of (1) being violated with a small probability by their local distribution, and (2) the two literals (corresponding to the other two variables $v,w$ in the constraint) are equal with at least $\Omega(1)$ probability. More precisely, we use the following definition.

\begin{definition}[Fixable variables]
    Let $d'$ be a positive integer such that $d' \ge 3$, let $W \subseteq V$ with $|W|\ge 3$, let $\rho \in \cD(W,d')$, and let $u \in W$. For any $\gamma_\unsat,\gamma_\fix,\gamma_\rate \in [0,1]$, the variable $u$ is $(\rho,\gamma_\unsat,\gamma_\fix,\gamma_\rate)$-fixable if
    \begin{equation*}
        \Pr_{\{v,w\} \sim \binom{W}{2}}
        \left[
            \begin{array}{c}
                |\{u,v,w\}|=3 \\
                \text{ \upshape and } \\
                \Pr_{\sigma \sim \rho_{\{u,v,w\}}}\left[P_{\{u,v,w\}}(\sigma)=0\right] \le \gamma_\unsat \\
                \text{ \upshape and } \\
                \Pr_{\beta \sim \rho_{\{v,w\}}}\left[P^v_{\{u,v,w\}}(\beta_v)=P^w_{\{u,v,w\}}(\beta_w)\right] \ge \gamma_\fix
            \end{array}
        \right]
        \ge \gamma_\rate \, .
    \end{equation*}
\end{definition}

\noindent
Intuitively, such a variable is fixable because if we sample a random pair of variables $v,w$ and a random assignment to them from their local distribution, the conditional pseudodistribution hence obtained should $(\gamma_\unsat/\gamma_\fix)$-fix the variable $u$, with probability $\gamma_\rate \cdot \gamma_\fix$. This should be the case because the constraint was almost satisfied to start with, and chances are that $v$ and $w$ will have the same literal, so we need to rely on the value of $u$ to keep the constraint satisfied. We refer the reader to \Cref{tab:fixing} for more intuition on the argument.

Having motivated this notion, we can exploit the fact that our instance is complete (and in particular any sub-instance $\cI[W]$ is complete) to show that there are in fact $\Omega(n)$ fixable variables, as stated by the following lemma. Again, one should think of $\val(\cI[W],\rho) = \tilde{O}(\opt)$.

\begin{lemma}
\label{lem:manyfixablevars}
Let $d'$ be a positive integer such that $d' \ge 3$, let $W \subseteq V$ with $|W|\ge 4$, let $\rho \in \cD(W,d')$, and let $\xi \in (0,1/4]$ such that $\val(\cI[W],\rho) \le \xi$. Then
    \begin{equation*}
        \left|  \left\{ u \in W: \, u \text{ \upshape is } \left(\rho, 2\xi,\frac{1}{12},\frac{1}{24}\right)\text{\upshape -fixable} \right\} \right| \ge  \frac{|W|}{24} \, .
    \end{equation*}
\end{lemma}
\begin{proof}
We consider the bipartite graph $G=(L \cup R, E)$ with bipartition $L,R$ and \smash{$L = \binom{W}{2}, R = W$}, where \smash{$E \subseteq \binom{W}{2} \times W$} is defined as follows: for each \smash{$\{v,w\} \in \binom{W}{2}, u \in W$}, the edge $(\{v,w\},u) $ is in $E$ if and only if the following condition holds:
\begin{equation*}
\begin{array}{c}
                |\{u,v,w\}|=3 \\
                \text{ \upshape and } \\
                \Pr_{\sigma \sim \rho_{\{u,v,w\}}}\left[P_{\{u,v,w\}}(\sigma)=0\right] \le 2 \cdot \xi \\
                \text{ \upshape and } \\
                \Pr_{\beta \sim \rho_{\{v,w\}}}\left[P^v_{\{u,v,w\}}(\beta_v)=P^w_{\{u,v,w\}}(\beta_w)\right] \ge \frac{1}{12} \, .
            \end{array}
\end{equation*}
Then, to prove the lemma statement we use the following claim.
\begin{claim}
    \label{claim:manyedges}
    $|E| \ge \binom{|W|}{3}/2$.
\end{claim}
\noindent
Before showing the claim, we check that it allows to prove the lemma. In fact, suppose for the sake of a contradiction that there is an $f$ fraction of $(\rho, 2\xi,{1}/{12},{1}/{24})$-fixable variables $u$ in $W$, with $f < 1/24$. Then, the number of edges in $G$ would be
\begin{align*}
    |E| & \le f  |W| \cdot  \binom{|W|}{2}+\left(1-f\right)|W| \cdot \frac{1}{24} \binom{|W|}{2} \\
    & < |W| \cdot \frac{1}{24} \binom{|W|}{2}+|W| \cdot \frac{23}{24^2} \binom{|W|}{2} \le \frac{1}{12}|W|\binom{|W|}{2}  \le \frac{1}{2} \binom{|W|}{3} \, ,
\end{align*}
contradicting \Cref{claim:manyedges}.

    \begin{proof}[Proof of \Cref{claim:manyedges}]
    Note that for any \smash{$\{u,v,w\} \in \binom{W}{3}$}, we can pick any \smash{$\{x,y\} \in \binom{\{u,v,w\}}{2}$} and accordingly let $z$ be the unique variable in $\{u,v,w\} \setminus \{x,y\}$, hence obtaining a valid pair $(\{x,y\},z)$  with $|\{x,y,z\}|=3$. Moreover, for at least \smash{$\binom{|W|}{3}/2$} of the subsets \smash{$\{u,v,w\} \in \binom{W}{3}$} one has the probability of constraint being violated is $ \Pr_{\sigma \sim \rho_{\{u,v,w\}}}[P_{\{u,v,w\}}(\sigma)=0] \le 2 \xi$. Also, for any \smash{$\{u,v,w\} \in \binom{W}{3}$} with \smash{$ \Pr_{\sigma \sim \rho_{\{u,v,w\}}}[P_{\{u,v,w\}}(\sigma)=0] \le 2 \xi$}, there exists some \smash{$\{x,y\} \in \binom{\{u,v,w\}}{2}$} such that $\Pr_{\beta \sim \rho_{\{x,y\}}}[P^x_{\{u,v,w\}}(\beta_x)=P^y_{\{u,v,w\}}(\beta_y)] \ge (1-2\xi)/(2^3-2)$, which implies $(\{x,y\},z) \in E$ for $z$ being the unique variable in $\{u,v,w\} \setminus \{x,y\}$. Hence, \smash{$|E| \ge \binom{|W|}{3}/2$}.
    \end{proof}
\noindent
    This concludes the proof of the lemma.
\end{proof}

\noindent
We are now at a point where we have convinced ourselves that fixable variables are useful, and that there are many of them. However, all the pairs $v,w$ that make a variable fixable may be different than those that make another variable fixable. We then introduce the notion of a \textit{fixing tuple}, that is a collection of few, $t$ pairs of variables $v,w$ (again, think $t=O(\log \log n)$) so that $\Omega(t)$ of them fix a given third variable $u$. Formally, we have the following definition.

\begin{definition}[Fixing tuple]
     Let $d',t$ be positive integers such that $d' \ge 3$, let $W \subseteq V$ with $|W|\ge 3$, let $\rho \in \cD(W,d')$, and let $\cT \in W^{2t}$ be a tuple of $t$ pairs $\cT=((v_1,w_1),\dots,(v_t,w_t))$ over $W$. For any $\gamma_\unsat,\gamma_\fix,\gamma_\rate' \in [0,1]$ and any \smash{$u \in {W}$}, the tuple $\cT$ is $(\rho,u,\gamma_\unsat,\gamma_\fix,\gamma_\rate')$-fixing if
    \begin{equation*}
        \Pr_{i \sim [t]}
        \left[
            \begin{array}{c}
                |\{u,v_i,w_i\}|=3 \\
                \text{ \upshape and } \\
                \Pr_{\sigma \sim \rho_{\{u,v_i,w_i\} }}\left[P_{\{u, v_i,w_i\}}(\sigma)=0\right] \le \gamma_\unsat \\
                \text{ \upshape and } \\
                \Pr_{\beta \sim \rho_{\{v_i,w_i\}}}\left[P^{v_i}_{\{u,v_i,w_i\}}(\beta_{v_i})=P^{w_i}_{\{u,v_i,w_i\}}(\beta_{w_i})\right] \ge \gamma_\fix
            \end{array}
        \right]
        \ge \gamma_\rate' \, .
    \end{equation*}
\end{definition}

\noindent
Intuitively, the goal would be to ensure for an $\Omega(1)$ fraction of variables $u$, a random collection of $t$ pairs contains at least one pair $v,w$ that fixes it. We in fact prove a stronger statement, saying that any given variable is fixed by $\Omega(t)$ pairs with $1-\eul^{-\Omega(t)}$  probability (which is $1-1/\polylog(n)$ if we plug in $t=O(\log \log n)$).

\begin{lemma}
\label{lem:fixingtuple}
    Let $d',t$ be positive integers such that $d' \ge 3$, let $W \subseteq V$ with $|W|\ge 3$, let $\rho \in \cD(W,d')$, let $\gamma_\unsat,\gamma_\fix,\gamma_\rate \in [0,1]$, and let $u \in W$ such that $u$ is  $(\rho,\gamma_\unsat,\gamma_\fix,\gamma_\rate)$-fixable. Then
    \begin{equation*}
        \Pr_{\cT \sim W^{2t}} \left[ \cT \text{ \upshape is $\left(\rho,u,\gamma_\unsat,\gamma_\fix,\frac{\gamma_\rate}{2} \right)$-fixing} \right] \ge 1-\eul^{-\gamma_\rate \cdot t/8}-\frac{t}{|W|}  \, .
    \end{equation*}
\end{lemma}
\begin{proof}
Let us identify each $\cT \in W^{2t}$ with a tuple of $t$ pairs $\cT=((v_1,w_1),\dots,(v_t,w_t))$ over $W$. If we sample $\cT \sim W^{2t}$, then the probability that there exists $i \in [t]$ such that $v_i=w_i$ is at most $t/|W|$.  Conditioning on this event and taking a union bound, we conclude the proof by bounding
    \begin{align*}
         \Pr_{\cT \sim \binom{W}{2}^t} \left[ \cT \text{ is not $\left(\rho,u,\gamma_\unsat,\gamma_\fix,\frac{\gamma_\rate}{2} \right)$-fixing}  \right] \le \eul^{-\gamma_\rate \cdot t/8} \, ,
    \end{align*}
    which follows from the Chernoff bound.
\end{proof}

\noindent
For technical reasons, it is useful to also introduce the notion of a \textit{fixing tuple-assignment pair}. For a given variable $u$, this is essentially a fixing tuple together with an assignment to its variables that witnesses its fixing capabilities with respect to $u$. Specifically, instead of just saying that for $\gamma_\rate' t$ pairs in the collection, the probability that the corresponding literals take the same value in the constraint with $u$ is $\gamma_\fix$, we write down an assignment to the $t$ pairs such that for $\gamma_\rate''=\gamma_\rate' \gamma_\fix t$ pairs in the collection, the corresponding literals take the same value in the constraint with $u$. This is formalized in the following definition.

\begin{definition}[Fixing tuple-assignment pair]
     Let $d',t$ be positive integers such that $d' \ge 3$, let $W \subseteq V$ with $|W|\ge 3$, let $\rho \in \cD(W,d')$, let $\cT \in W^{2t}$ be a tuple of $t$ pairs $\cT=((v_1,w_1),\dots,(v_t,w_t))$ over $W$, and let $\beta \in \{0,1\}^T$. For any $\gamma_\unsat,\gamma_\rate'' \in [0,1]$ and any $u\in W$, the tuple-assignment pair $(\cT,\beta)$ is $(\rho,u,\gamma_\unsat,\gamma_\rate'')$-fixing if
    \begin{equation*}
        \Pr_{i \sim [t]}
        \left[
            \begin{array}{c}
                |\{u,v_i,w_i\}|=k \\
                \text{ \upshape and } \\
                \Pr_{\sigma \sim \rho_{\{u,v_i,w_i\} }}\left[P_{\{u,v_i,w_i\}}(\sigma)=0\right] \le \gamma_\unsat \\
                \text{ \upshape and } \\
                P^{v_i}_{\{u,v_i,w_i\}}(\beta_{v_i})=P^{w_i}_{\{u,v_i,w_i\}}(\beta_{w_i})
            \end{array}
        \right]
        \ge \gamma_\rate'' \, .
    \end{equation*}
\end{definition}

\noindent
Given a fixing tuple, we can attach it with an assignment to obtain a fixing tuple-assignment pair by sampling the assignment to each pair $v,w$ in the tuple independently. A Chernoff bound would give a success probability of $1-\eul^{-\Omega(t)}$ (which again is $1-1/\polylog(n)$ if $t=O(\log \log n)$). However, we can only afford to sampled one assignment from the pseudodistribution we are given, for reasons mentioned in \Cref{subsec:techoverview_algo}. This is the point where we assume that the probability argument of \Cref{lem:tvd} holds: this allows us to sample a collective assignment from the local distribution of the collection of variables, as opposed to sampling from the local distribution of each of its pairs independently.

\begin{lemma}
\label{lem:fixingpair}
    Let $d',t$ be positive integers such that $d' \ge 2t+3$, let $W \subseteq V$ with $|W|\ge 3$, let $\cT \in W^{2t}$ such that $|T|=2t$, let $\rho \in \cD(W,d')$, let $\eta \in [0,1]$ such that  \smash{$\|\rho_T - \otimes_{v \in T}\rho_{\{v\}}\|_1 \le \eta$}, let $\gamma_\unsat,\gamma_\fix,\gamma_\rate' \in [0,1]$, and let \smash{$u \in W$} such that $\cT$ is $(\rho,u,\gamma_\unsat,\gamma_\fix,\gamma_\rate')$-fixing. Then
    \begin{equation*}
        \Pr_{\beta \sim \rho_T} \left[ (\cT,\beta) \text{ \upshape is $\left(\rho,u,\gamma_\unsat,\frac{\gamma_\rate' \cdot \gamma_\fix^2}{8}\right)$-fixing} \right] \ge 1  - \eul^{-\frac{\gamma_\fix^2}{32}\gamma_\rate' \cdot t}-\eta \, .
    \end{equation*}
\end{lemma}
\begin{proof}
    Let us identify $\cT$ with a tuple of $t$ pairs over $W$ as $\cT=((v_1,w_1),\dots,(v_t,w_t))$. Since we assume that $\cT$ is $(\rho,u,\gamma_\unsat,\gamma_\fix,{\gamma_\rate}')$-fixing, we know that there is a set $I \subseteq [t]$ with $|I| \ge \gamma_\rate'  \cdot t$ such that for every $i \in I$ one has $|\{u,v_i,w_i\}|=k$,  $\Pr_{\sigma \sim \rho_{\{u,v_i,w_i\}}}[P_{\{u,v_i,w_i\}}(\sigma)=0] \le \gamma_\unsat$, and \smash{$\Pr_{\beta \sim \rho_{\{v_i,w_i\}}}[P^{v_i}_{\{u,v_i,w_i\}}(\beta_{v_i})=P^{w_i}_{\{u,v_i,w_i\}}(\beta_{w_i})] \ge \gamma_\fix$}. Hence, for each $i \in I$ there exists $b(i) \in \{0,1\}$ such that
    \begin{equation*}
        \Pr_{\beta \sim \rho_{\{v_i,w_i\}}}\left[(P^{v_i}_{\{u,v_i,w_i\}}(\beta_{v_i}), P^{w_i}_{\{u,v_i,w_i\}}(\beta_{w_i}))=(b(i),b(i))\right] \ge \frac{\gamma_\fix}{2} \, .
    \end{equation*}
    Then, observe that for each $i \in I$ one has
    \begin{equation*}
        \Pr_{\beta \sim \rho_{\{v_i\}} \otimes \rho_{\{w_i\}}}\left[P^{v_i}_{\{u,v_i,w_i\}}(\beta_{v_i})= P^{w_i}_{\{u,v_i,w_i\}}(\beta_{w_i})\right] \ge \frac{\gamma_\fix^2}{4} \, .
    \end{equation*}
    Together with the Chernoff bound, this gives
    \begin{equation*}
        \Pr_{\beta \sim \otimes_{i \in I} \left(\rho_{\{v_i\}} \otimes \rho_{\{w_i\}}\right)}\left[\left|\left\{i \in I: \, P^{v_i}_{\{u,v_i,w_i\}}(\beta_{v_i})= P^{w_i}_{\{u,v_i,w_i\}}(\beta_{w_i})\right\}\right| < \frac{\gamma_\fix^2}{8}|I|\right] \le   \eul^{-\frac{\gamma_\fix^2}{32}|I|}\, .
    \end{equation*}
    Recalling that $|I| \ge \gamma_\rate' \cdot t$ and using \smash{$\|\rho_T - \otimes_{v \in T}\rho_{\{v\}}\|_1 \le \eta$} we get the claim.
\end{proof}

\noindent
Now we know that a given a fixing tuple, we can obtain a fixing tuple-assignment pair with good probability. We are now in shape to formalize intuition sketched at the beginning of this section: if a tuple-assignment pair $\cT,\beta$ is fixing for a variable $u$ with good probability, then with good probability $u$ is almost fixed in the pseudodistribution obtained by conditioning on the variables in the tuple taking the assignment $\beta$.

\begin{lemma}
\label{lem:fixablegetsnearlyfixed}
     Let $d',t$ be positive integers such that $d' \ge 2t+3$, let $W \subseteq V$ with $|W|\ge 3$, let $\cT \in W^{2t}$ such that $|T|=2t$, let $\rho \in \cD(W,d')$, let $\gamma_\unsat,\gamma_\fix,\gamma_\rate'',\zeta \in [0,1]$, and let  \smash{$u \in W$} such that $\Pr_{\beta \sim \rho_T} [ (\cT,\beta) \text{ \upshape is $(\rho,u,\gamma_\unsat,\gamma''_\rate)$-fixing}] \ge 1-\zeta$. Then, for any $\tau >0$ one has
    \begin{equation*}
        \Pr_{\beta \sim \rho_T} \left[ \min_{b \in \{0,1\}} \,\, \Pr_{\alpha \sim (\rho|_{T \gets \beta})_{\{u\}}}\left[\alpha=b\right] > \tau \cdot \gamma_\unsat \right] \le \frac{1}{(1-\zeta)\tau \gamma''_\rate } +\zeta \, .
    \end{equation*}
\end{lemma}
\begin{proof}
    Again, let us identify $\cT \in W^{2t}$ with a tuple of $t$ pairs $\cT=((v_1,w_1),\dots,(v_t,w_t))$ over~$W$, and let $I \subseteq [t]$ be the set of $i \in [t]$ such that $|\{u,v_i,w_i\}|=3$ and $\Pr_{\sigma \sim \rho_{\{u,v_i,w_i\}}}[P_{\{u,v_i,w_i\}}(\sigma)=0] \le \gamma_\unsat$. Now, on the one hand we have
    \begin{align}
    \label{eq:ub}
           \sum_{i \in I}\Pr_{\sigma \sim \rho_{\{u,v_i,w_i\}}}\left[P_{\{u,v_i,w_i\}}(\sigma)=0 \right] \le \gamma_\unsat \cdot t \, ,
    \end{align}
    and also
    \begin{align}
    \label{eq:conditioning}
          \ex_{\beta \sim \rho_T} \left[ \Pr_{\sigma \sim (\rho|_{T \gets \beta})_{\{u,v_i,w_i\}}}\left[P_{\{u,v_i,w_i\}}(\sigma)=0 \right] \right] = \Pr_{\sigma \sim \rho_{\{u,v_i,w_i\}}}\left[P_{\{u,v_i,w_i\}}(\sigma)=0 \right]  \, .
    \end{align}
    On the other hand, for any $\beta \in \supp(\rho_T)$ and any $i \in [t]$ such that $P^{v_i}_{\{u,v_i,w_i\}}(\beta_{v_i})=P^{w_i}_{\{u,v_i,w_i\}}(\beta_{w_i})$, we have that there is \smash{$\hat{b}(\beta,i) \in \{0,1\}$} satisfying 
    \begin{equation*}
        \left(P^{u}_{\{u,v_i,w_i\}}( \hat{b}(\beta,i) ), \, P^{v_i}_{\{u,v_i,w_i\}}(\beta_{v_i}), \, P^{w_i}_{\{u,v_i,w_i\}}(\beta_{w_i})\right) = (c,c,c) \quad \text{for some $c \in \{0,1\}$.}
    \end{equation*}
    This implies that for any $\beta \in \supp(\rho_T)$ and any $i \in [t]$ such that $P^{v_i}_{\{u,v_i,w_i\}}(\beta_{v_i})=P^{w_i}_{\{u,v_i,w_i\}}(\beta_{w_i})$, we have
    \begin{align*}
        \min_{b \in \{0,1\}} \, \, \Pr_{\alpha \sim (\rho|_{T \gets \beta})_{\{u\}}}\left[\alpha=b\right] & \le
        \Pr_{\alpha \sim (\rho|_{T \gets \beta})_{\{u\}}}\left[\alpha=\hat{b}(\beta,i)\right] \\
        & = \Pr_{\sigma \sim (\rho|_{T \gets \beta})_{\{u, v_i,w_i\}}}\left[\sigma_u=\hat{b}(\beta,i)\right] \\
        & \le \Pr_{\sigma \sim (\rho|_{T \gets \beta})_{\{u, v_i,w_i\}}}\left[P_{\{u, v_i,w_i\}}(\sigma)=0\right] \, .
    \end{align*}
    In particular, if we assume that $(\cT,\beta)$ is $(\rho,u,\gamma_\unsat,\gamma''_\rate)$-fixing, we know that there are at least $\gamma''_\rate \cdot t$ indices $i \in [t]$ for which the above bound holds, since there are at least $\gamma''_\rate \cdot t$ indices $i \in [t]$ such that \smash{$P^{v_i}_{\{u,v_i,w_i\}}(\beta_{v_i})=P^{w_i}_{\{u,v_i,w_i\}}(\beta_{w_i})$}. Hence, conditioning on $(\cT,\beta)$ being \smash{$(\rho,u,\gamma_\unsat,\gamma''_\rate)$}-fixing and combining with observations~\eqref{eq:ub} and~\eqref{eq:conditioning} we get
    \begin{equation*}
        (1-\zeta)\gamma''_\rate t \cdot \ex_{\beta \sim \rho_T }\left[\min_{b \in \{0,1\}} \,\, \Pr_{\alpha \sim (\rho|_{T \gets \beta})_{\{u\}}}\left[\alpha=b\right] \, \middle| \, \text{$(\cT,\beta)$ is $(\rho,u,\gamma_\unsat,\gamma''_\rate)$-fixing} \right] \le \gamma_\unsat  t \, .
    \end{equation*}
    \noindent
    The claim then follows by conditioning and Markov's inequality with a union bound.
\end{proof}

\noindent
We can finally show the main lemma of the section.

\begin{proof}[Proof of \Cref{lem:condtofix}]
    Hereafter we fix $r' \in \{0,1,\dots,r\}$ such that
    \begin{equation*}
        \Pr_{\substack{\cS \sim W^{r'}, \, \sigma \sim  \rho_{S}, \\ \cT \sim W^{2t}}} \left[\left\|(\rho|_{S \gets \sigma})_T - \otimes_{v \in T}(\rho|_{S \gets \sigma})_{\{v\}}\right\|_1 > 2\cdot \frac{3^{2t/3}}{ r^{1/3}}\right] \le \frac{3^{2t/3}}{ r^{1/3}} \, ,
    \end{equation*}
    which is guaranteed to exist by \Cref{lem:tvd}.  For convenience, we set
    \begin{equation*}
        \tau'=\sqrt{\tau}/2 \, ,\quad \xi= \tau' \cdot \val(\cI[W],\mu)\, , \quad  \gamma_\unsat=2\xi  \, , \quad \gamma_\fix=\frac{1}{12}\, , \quad \gamma_\rate =\frac{1}{24} \, ,
    \end{equation*}
    and
    \begin{equation*}
         \eta =  2\cdot \frac{3^{2t/3}}{ r^{1/3}} \, , \quad \chi=\eul^{-\gamma_\rate \cdot t/8}+\frac{t}{|W|}  \, , \quad \zeta=\eul^{-\frac{\gamma_\fix^2}{64}\gamma_\rate \cdot t}+\eta\, .
    \end{equation*}
    For any $ \cS \in W^{r'}, \, \sigma \in \supp(\rho_S)$, we define
    \begin{equation*}
        \rho(S,\sigma) = \mu|_{S \gets \sigma} \, ,
    \end{equation*}
    and also define the set of fixable variables
    \begin{equation*}
        X(S,\sigma) = \left\{u \in W : \, u \text{ is $\left(\rho(S,\sigma),\gamma_\unsat,\gamma_\fix,\gamma_\rate\right)$-fixable} \right\} \, .
    \end{equation*}
    For any $ \cS \in W^{r'}, \, \sigma \in \supp(\rho_S), \, \cT \in W^{2t}$ we define the of variables for which $\cT$ is fixing
    \begin{equation*}
        X^*(S,\sigma,\cT) = \left\{u \in X(S,\sigma) : \, \cT \text{ is $\left(\rho(S,\sigma),u,\gamma_\unsat,\gamma_\fix,\frac{\gamma_\rate}{2}\right)$-fixing} \right\} \, .
    \end{equation*}
    With this notation in place we can easily prove the following claim, which establishes sufficient conditions for the random choice of $\cS,\sigma,\cT,\beta$ to get a large set $F_W(\rho(S,\sigma)|_{T \gets \beta}, \, \tau'  \gamma_\unsat )$.

    \begin{claim}
        \label{claim:fixedchoice}
        Let $ \cS \in W^{r'}, \, \sigma \in \supp(\rho_S)$ such that
    \begin{equation*}
        \val(\cI[W],\rho(S,\sigma)) \le \tau' \cdot \val(\cI[W],\mu)\, .
    \end{equation*}
    Let $ \cT \in W^{2t}$ such that
    \begin{equation*}
        | X^*(S,\sigma,\cT)| \ge \frac{1}{2}|X(S,\sigma)| \, , \quad |T|=2t \, , \quad \left\|(\mu|_{S \gets \sigma})_T - \otimes_{v \in T}(\mu|_{S \gets \sigma})_{\{v\}}\right\|_1 \le \eta \, .
    \end{equation*}
    Let $\beta \in \supp(\rho(S,\sigma)_T)$ such that
    \begin{equation*}
        \left|\left\{u \in X^*(S,\sigma,\cT): \,  \min_{b \in \{0,1\}} \,\, \Pr_{\alpha \sim (\rho(S,\sigma)|_{T \gets \beta})_{\{u\}}}\left[\alpha=b\right] > \tau' \cdot \gamma_\unsat \right\}\right| \le \frac{1}{2} \cdot \left|X^*(S,\sigma,\cT)\right| \, .
    \end{equation*}
    Then,
    \begin{align*}
         \left|F_W(\rho(S,\sigma)|_{T \gets \beta}, \, \tau' \cdot \gamma_\unsat )\right| \ge  \frac{1}{96} \left| W\right| \, .
    \end{align*}
    \end{claim}
    \begin{proof}
         Together with our assumption that $\val(\cI[W],\mu) \le 1/(3\tau) $, the requirement on $\cS,\sigma$ implies
    \begin{equation*}
        \val(\cI[W],\rho(S,\sigma)) \le \xi  \le \frac{1}{4} \, .
    \end{equation*}
    Thus, we can employ \Cref{lem:manyfixablevars} with $\rho=\rho(S,\sigma)$, which gives
    \begin{equation*}
        |X(S,\sigma)| \ge \frac{|W|}{24} \, .
    \end{equation*}
    Combining this with the requirement on $\cT,\beta$ we get that

    \begin{align*}
         \left|F_W(\rho(S,\sigma)|_{T \gets \beta}, \, \tau' \cdot \gamma_\unsat )\right| & \ge \frac{1}{2} \left| X^*(S,\sigma,\cT)\right| \\
         & \ge \frac{1}{4} \left| X(S,\sigma)\right| \\
         & \ge  \frac{1}{96} \left| W\right| \, .
    \end{align*}
    \end{proof}

   \noindent
    To conclude the lemma statement, all is left to do is showing that the above conditions on the choice of $\cS,\sigma,\cT,\beta$ in \Cref{claim:fixedchoice} hold with high probability. First, by Markov's inequality we have
    
    \begin{equation*}
    \Pr_{\cS \sim W^{r'}, \, \sigma \sim  \rho_{S}}\left[\val(\cI[W],\rho(S,\sigma)) > \tau' \cdot \val(\cI[W],\mu)\right] \le \frac{1}{\tau'} \, .
    \end{equation*}
    By our choice of $r'$ and definition of $\eta$ we also have
    \begin{equation*}
        \Pr_{\substack{\cS \sim W^{r'}, \, \sigma \sim  \rho_{S}, \\ \cT \sim W^{2t}}} \left[\left\|(\rho|_{S \gets \sigma})_T - \otimes_{v \in T}(\rho|_{S \gets \sigma})_{\{v\}}\right\|_1 > \eta \right] \le \frac{1}{2}\eta \, .
    \end{equation*}
    By an elementary bound, we also get
    \begin{equation*}
        \Pr_{\substack{ \cT \sim W^{2t}}} \left[|T| <2t \right] \le \frac{4t^2}{|W|} \, .
    \end{equation*}
    Using \Cref{lem:fixingtuple} with $d'=d-r' \ge 3, \, \rho=\rho(S,\sigma) \in \cD(W,d-r')$ together with Markov's inequality, we obtain
    \begin{equation*}
        \Pr_{\substack{\cS \sim W^{r'}, \, \sigma \sim  \rho_{S}, \\ \cT \sim W^{2t}}} \left[ |X(S,\sigma) \setminus X^*(S,\sigma,\cT)| > \frac{1}{2}|X(S,\sigma)|\right] \le 2\chi \, .
    \end{equation*}
    Finally, if we condition on $\cS,\sigma,\cT$ that satisfy these conditions, we can employ \Cref{lem:fixingpair} and \Cref{lem:fixablegetsnearlyfixed} with $d'=d-r' \ge 2t+3, \, \rho=\rho(S,\sigma) \in \cD(W,d-r'),\, \gamma_\rate'=\gamma_\rate/2$ and apply Markov's inequality to the latter;  then taking a union bound gives that the failure probability is bounded by
    \begin{align*}
        & \frac{1}{\tau'}+\frac{1}{2}\eta+\frac{4t^2}{|W|}+2\chi+2 \cdot \left(\frac{1}{(1-\zeta)\tau' \cdot \left( \gamma'_\rate \gamma_\fix^2/8\right) } +\zeta\right) \\
        \le & \frac{10^6}{\tau'}+5\cdot \frac{3^{2t/3}}{ r^{1/3}}+\frac{5t^2}{|W|}+3\cdot \eul^{-t/10^6} \, .
    \end{align*}
\end{proof}

	\section{Proof of the Thresholding Lemma}\label{sec:threshround}
        
        In this section, we prove \Cref{lemma:threshround}, restated below for convenience. The goal is to analyze the change in the LP value from the thresholding procedure of line \ref{step:threshround} in \Cref{alg:roundingalgo}.

        \threshround*

\noindent
        We seek to bound the increase in the aggregate value after rounding the $\theta$-fixed variables for $\theta\sim[\tau\delta, 2\tau\delta]$. The analysis considers each constraint individually. Each constraint, depending on which of the four classes $\cP_0\{V_\rU\},\cP_1\{V_\rU\},\cP_2\{V_\rU\},\cP_3\{V_\rU\}$ it belongs to, contributes to one of $\lp_0^{V_\rU}(\rho)$, $\lp_1^{V_\rU}(\rho)$, $\lp_2^{V_\rU}(\rho)$, $\lp_3^{V_\rU}(\rho)$ and hence has a different contribution to $A^{V_\rU}(\rho)$ (because the latter is a weighted sum). As we fix some variables to their ruling values, the constraints might change class, going from $\cP_j\{V_\rU\}$ to $\cP_i\{V_\rU'\}$ for some $i\leq j$, where $V_\rU'$ is the updated set of unfixed variables after rounding. As a consequence, the contribution of the constraint to $A^{V_\rU'}(\rho')$ might also be different. In the following proof, we show that for each of the four types of constraints, the new contribution is not much more than the starting one.
    
        \begin{proof}
            Recall that the aggregate value is defined as
            $$A^{V_\rU}(\rho)= \tau \cdot \log^3 n \cdot \lp^{V_\rU}_3(\rho) + \log^2 n \cdot \lp^{V_\rU}_2(\rho) + \log n \cdot \lp^{V_\rU}_1(\rho) + \lp^{V_\rU}_0(\rho) \, .$$
            We will refer to a constraint $P_S \in \cP_i\{V_\rU\}$ as an $i$-constraint, for each $i\in\{0,1,2,3\}$. Moreover, we adopt some short-hand notation for convenience, so for any $\theta \in [0,1/2)$ we call:
            \begin{itemize}
            	\item the set of $(\rho,\theta)$-fixed variables $F^\theta=F_{V_\rU}(\rho, \theta)$ and their ruling assignment $\beta^\theta=\omega^*_{V_\rU}(\rho,\theta)$;
            	\item the updated set of unfixed variables $V_\rU^\theta = V_\rU \setminus F^\theta$;
            	\item the pseudodistribution after fixing the $(\rho,\theta)$-fixed variables $\rho^\theta = \rho^{F^\theta\gets \beta^\theta}$;
            	\item the LP values  $\lp_i = \lp^{V_\rU}_i(\rho)$ and $\lp^\theta_i=\lp^{V_\rU^\theta}_i(\rho^\theta)$, for each $i\in\{0,1,2,3\}$;
            	\item the aggregate values $A(\rho)=A_\tau^{V_\rU}(\rho)$ and $A^\theta(\rho^\theta)=A_\tau^{V_\rU^\theta}(\rho^\theta)$;
                \item a constraint $P\in \cP_i\{V_U\}$ an $i$-constraint.
            \end{itemize}
            
            \noindent 
            The idea is to analyze the contribution to the $\lp^\theta_i$ from all the $j$-constraints. Such a constraint contributes to $\lp^\theta_i$ if and only if exactly $j-i$ variables are rounded, so that after rounding it moves from \smash{$\cP_j\{V_\rU\}$ to $\cP_i\{V_\rU^\theta\}$}. To do this, let us define $\Delta^\theta_{ji}$ as the contribution to $\lp_i^\theta$ from the $j$-constraints that end up in \smash{$\cP_i\{V_\rU^\theta\}$}, for all $i \le j$. More precisely, we let \begin{equation*}
                \Delta^\theta_{j,i}=\sum_{P_S \in \cP_i\{V_\rU^\theta\} \cap \cP_{j}\{V_\rU\}} \Pr_{\alpha \sim \rho^\theta}[P_S(\alpha)=0] \, ,
            \end{equation*}
            for all $i,j \in \{0,1,2,3\}$ with $i \le j$. Then, for all $i\in\{0,1,2,3\}$, we have
            $$\lp_i^\theta = \sum_{j=i}^{3} \Delta_{j,i} \, .$$
            \noindent
            In what follows, whenever $\theta$ is clear from context, we allow ourselves to simply the notation further as follows: for $F^\theta$, $\beta^\theta$ and $ \Delta^{\theta}_{ji}$  we drop the superscript, while for $V_\rU^\theta,\rho^\theta, \lp^\theta_i, A^\theta(\rho^\theta)$ we replace the superscript with a dash as  $V_\rU',\rho', \lp', A'(\rho')$. 
            
            Now, we start to individually bound the values of $\Delta^\theta_{j,i}$. Note that a trivial upper bound is $\Delta^\theta_{i,i} \leq \lp_{i}$, for all $i \in \{0,1,2,3\}$.
            
            \begin{claim}\label{lem:10}
                Let $\theta \in [0,1/5]$. Then, one has $\Delta^\theta_{1,0}\leq 2\cdot\lp_1$.
            \end{claim}
            \begin{proof}
                First note that any contribution to $\Delta_{1,0}$ can only come from the constraints $P_S$ in $\cP_1\{V_\rU\}$ each with exactly one variable in $V_\rU$ which gets rounded (which causes $P_S$ to move to $\cP_0\{V_\rU'\}$).
                
                Consider any such $1$-constraint $P_S$, where $v \in S$ is the only variable in $V_\rU$. It must then be the case that $v$ is $(\rho,\theta)$-fixed, i.e., either $\rho_{\{v\}}(0) \le \theta$ or $\rho_{\{v\}}(1) \le \theta$. Since $S\setminus \{v\} \subseteq V\setminus V_\rU$, this implies that either
                $$\Pr_{\alpha\sim\rho_{S}}[P_{S}(\alpha)=0]\leq \theta \quad \text{ or } \quad\Pr_{\alpha\sim\rho_{S}}[P_{S}(\alpha)=1]\leq \theta \, . $$
                The former case is trivial as it means that the ruling value of $v$ together with the already determined assignment to of $S\setminus \{v\}$ leads to an assignment that satisfies $P_S$, and thus contributes $0$ to $\lp'_0$.
                For the other case, it implies that the contribution of $P_{S}$ to $\lp_1$  is  $\Pr_{\alpha\sim\rho_{S}}[P_{S}(\alpha)=0] \ge 1-\theta$. Note that the contribution of $P_S$ to $\lp'_0$ is always upper bounded $1$, which is at most $ \frac{1}{1-\theta} \Pr_{\alpha\sim\rho_{S}}[P_{S}(\alpha)=0]$.  Thus, the total contribution of all the $1$-constraints to $\lp'_0$ is
                $$\Delta_{1,0}\leq\frac{1}{1-\theta}\lp_1\leq 2\cdot\lp_1\, .$$

                \end{proof}

                \begin{claim}
                \label{lem:11}
                    Let $\theta \in [0,1/5]$. Then, one has $\Delta^\theta_{2,0}\leq 2\cdot\lp_1$.
                \end{claim}
                \begin{proof}
                Similar to \Cref{lem:10}, for the analysis of $\Delta_{2,0}$ we individually consider the contribution from every constraint in $\cP_2\{V_\rU\}$ that ends up in $ \cP_0\{V_\rU'\}$, i.e., the $2$-constraints where both unfixed variables get rounded.
                
                Consider a $2$-constraint $P_{S}=P_{\{x,y,z\}}$, with $x,y \in V_\rU \cap V_\rU'$ and $z \notin V_\rU$. Let us call $b_x, b_y\in\{0,1\}$ the ruling values of $x$ and $y$ respectively, and also let $b_z \in \{0,1\}$ be such that $\rho_{\{z\}}(b_z)=1$. In the case when $P_S^x(b_x)\neq P_S^z(b_z)$ or $P_S^y(b_y)\neq P_S^z(b_z)$ we obtain an assignment that satisfies $P_S$, so the constraint does not contribute anything to $\lp'_0$. Hence, let us focus on the case when $P_S^x(b_x)= P_S^y(b_y)= P_S^z(b_z)$ and hence $P_S$ contributes $1$ to $\lp'_0$. We know that both $x$ and $y$ are $(\rho,\theta)$-fixed and therefore
                
                $$\Pr_{\alpha \sim \rho_{\{x\}}}[\alpha=1-b_x]\le \theta \quad \text{ and } \quad  \Pr_{\alpha\sim \rho_{\{y\}}}[\alpha=1-b_y] \le \theta \, . $$
                Using the union bound, this implies that the contribution of $P_S$ to $\lp_2$ can be lower-bounded as                \begin{align*}
                    \Pr_{\alpha \sim\rho_{S}}[P_{S}(\alpha)=0] = & 1-\Pr_{\alpha \sim\rho_{\{x,y\}}}[\alpha_x=1-b_x \, \mbox{ or } \, \alpha_y=1-b_y] \\
                    & \ge  1 - \Pr_{\alpha\sim \rho_{\{x\}}}[\alpha=1-b_x]-\Pr_{\alpha\sim \rho_{\{y\}}}[\alpha=1-b_y] \\
                    & \ge 1-2\theta
                \end{align*}
\noindent
                Thus, the total contribution from all such constraints is upper-bounded by
                $$\Delta_{2,0}\leq \frac{1}{1-2\theta}\lp_2\leq 2\cdot\lp_2\, . $$

            \end{proof}

            \begin{claim}\label{lem:3round}
                Let $\theta \in [0,2\tau\delta]$. Then, one has $\Delta^\theta_{3,i}\leq \lp_3 + 6\tau\delta \binom{|V_\rU|}{3}$ for all $i\in\{0,1,2\}$.
            \end{claim}
            \begin{proof}

                The idea here is that if we round a $(\rho,\theta)$-fixed variable then the contribution of its clauses after rounding it is increased by at most an additive term of $\theta$.

                More formally, consider a constraint $P_S \in \cP_3\{V_\rU\} \cap \cP_i\{V_\rU'\}$, which contributes  $\Pr_{\alpha \sim \rho_S}[P_S(\alpha)=0]$ to $\lp_3$. Let \smash{$\hat{S}= S\cap F$} be the $(\rho,\theta)$-fixed variables of $S$ that we round to their ruling assignment \smash{$\hat{\beta} = \beta_{\hat{S}}$}. Note that by assumption $|\hat{S}|=3-i$ with $i \in \{0,1,2\}$. Also, for any \smash{$\alpha\in \{0,1\}^{S\setminus \hat{S}}$}, let us denote by \smash{$\alpha\oplus \hat{\beta} \in \{0,1\}^S$} the vector assigning \smash{$\hat{\beta}$ to $\hat{S}$} and $\alpha$ to $S\setminus \hat{S}$. With this notation, we want to upper-bounded the contribution of $P_S$ to $\lp'_i$. First, observe that
                \begin{equation*}
                	\Pr_{\alpha\sim \rho'_{S}}[P_{S}(\alpha)=0] =\Pr_{\alpha\sim \rho'_{S\setminus\hat{S}}}[P_{S}(\alpha\oplus \hat{\beta})=0]  =\Pr_{\alpha \sim \rho_{S}}[P_{S}(\alpha_{S\setminus\hat{S}}\oplus \hat{\beta})=0]  \, ,
                \end{equation*}
                since the distribution $\rho'$ obtained after fixing preserves the marginals of the unfixed variables. Now, with basic calculations we get
                \begin{align*}
                	\Pr_{\alpha \sim \rho_{S}}[P_{S}(\alpha_{S\setminus\hat{S}}\oplus \hat{\beta})=0] & = \Pr_{\alpha \sim \rho_{S}}[P_{S}(\alpha_{S\setminus\hat{S}}\oplus \hat{\beta})=0, \, \alpha_{\hat{S}}=\hat{\beta}] + \Pr_{\alpha \sim \rho_{S}}[P_{S}(\alpha_{S\setminus\hat{S}}\oplus \hat{\beta})=0, \, \alpha_{\hat{S}}\neq\hat{\beta}] \\
                	& \le \Pr_{\alpha \sim \rho_{S}}[P_{S}(\alpha)=0] + \Pr_{\alpha \sim \rho_{S}}[ \alpha_{\hat{S}}\neq\hat{\beta}] \\
                	& \le \Pr_{\alpha \sim \rho_{S}}[P_{S}(\alpha)=0] + \sum_{v \in \hat{S}} \Pr_{\alpha \sim \rho_{\{v\}}}[ \alpha_{v}\neq\hat{\beta}_v]  \, .
                \end{align*}
                Combining the above bounds and recalling that each $v \in \hat{S}$ is $(\rho,\theta)$-fixed and that $\hat{\beta}_v$ is its ruling value we get
                \begin{equation*}
						\Pr_{\alpha\sim \rho'_{S}}[P_{S}(\alpha)=0] \le \Pr_{\alpha \sim \rho_{S}}[P_{S}(\alpha)=0] + 3\theta \, .
                \end{equation*}
    			\noindent
                Thus, we have that each $3$-constraint in $\cP_3\{V_\rU\}$, where $i\in\{0,1,2\}$ variables remain unfixed, contributes to $\lp'_i$ at most its contribution to $\lp_3$ plus $3\theta$. Thus, we get
                \begin{equation*}
                	\Delta_{3,i} \le \lp_3+3\theta \cdot \binom{|V_\rU|}{3} \, .
                \end{equation*}
            \end{proof}

            \noindent
            One interesting thing to notice from this proof is that the logic of the additive $O(\theta)$ contribution applies to constraints from any of $\cP_1\{V_\rU\},\cP_2\{V_\rU\},\cP_3\{V_\rU\}$. However, we would not get our desired bounds if we have an additive increase of $O(\theta)$ for each constraint form $\cP_1\{V_\rU\},$ or $\cP_2\{V_\rU\}$, as there need not be at most $O(|V_\rU|^3)$ many such clauses. For instance, the maximum number of $2$-constraints could be $|V\setminus V_\rU|\cdot \binom{|V_\rU|}{2}$. We fix this by randomizing the threshold $\theta$ and looking at the expected contribution of these clauses.
            
            \begin{claim}\label{lem:20round}
                One has $\E_{\theta \sim [\tau\delta, 2\tau\delta]}[\Delta^\theta_{2,1}]\leq 3\lp_2$.
            \end{claim}
            \begin{proof}
                Expanding out the definition of $\Delta^\theta_{2,1}$ and using linearity of expectation, and the definition of $V_\rU^\theta$, we can rewrite
                \begin{align*}
                    \E_{\theta \sim [\tau\delta, 2\tau\delta]}\left[\Delta^\theta_{2,1}\right] = \sum_{P_S \in \cP_2\{V_\rU\}} \E_{\theta \sim [\tau\delta, 2\tau\delta]}\left[\mathbbm{1}_{P_S \in \cP_1\{V_\rU^\theta\}} \cdot \Pr_{\alpha \in \rho^\theta_S}[P_S(\alpha)=0]\right]  \, .
                \end{align*}
                For every $v \in V_\rU$, we define for convenience \smash{$\xi_v=\min_{b \in \{0,1\}} \rho_{\{w\}}(b)$}, and for any $2$-constraint $P_{S}\in\cP_2\{V_\rU\}$, we let \smash{$x^S,y^S \in S\cap V_\rU$} be the two unfixed variables in $S$ such that $\xi_{x^S} \le \xi_{y^S}$ and let $z^S \in S\setminus V_\rU$ be the fixed variables in $S$. Then, for any $\theta \in [\tau\delta, 2\tau\delta]$ and any $P_{S}\in\cP_2\{V_\rU\}$, we have
                \begin{equation*}
                    \mathbbm{1}_{P_S \in \cP_1\{V_\rU^\theta\}} \cdot \Pr_{\alpha \in \rho^\theta_S}[P_S(\alpha)=0] = \mathbbm{1}_{\theta \in [\xi_{x^S},\xi_{y^S})} \cdot \Pr_{\alpha \in \rho^\theta_S}[P_S(\alpha)=0] \, .
                \end{equation*}
                Fix for now $\theta$ and $P_S$, and let us assume that \smash{$\theta \in [\xi_{x^S},\xi_{y^S})$}, so we have $x \notin V_\rU^\theta, y \in V_\rU^\theta$. Under this assumption, we want to bound the contribution of $P_S$ to $\lp_1^\theta$. Similarly to \Cref{lem:3round}, this contribution can be larger than the contribution to $\lp_2$ by no more than an additive term of $\theta$. More precisely, if we let $\hat{\beta}_{x^S} \in \{0,1\}$ be the ruling value of $x^S$, calculations similar to those in the proof of the proof of \Cref{lem:3round} give
                \begin{align*}
                    \Pr_{\alpha \sim \rho^\theta_{S}}[P_{S}(\alpha) = 0] & \le \Pr_{\alpha \sim \rho_{S}}[P_{S}(\alpha) = 0]+\Pr_{\alpha \sim \rho_{\{x^S\}}}[\alpha \neq \hat{\beta}_{x^S}] \\
                    & = \Pr_{\alpha \sim \rho_{S}}[P_{S}(\alpha) = 0]+\xi_{x^S} \\
                    & \le \Pr_{\alpha \sim \rho_{S}}[P_{S}(\alpha) = 0]+\theta \\
                    & \le \Pr_{\alpha \sim \rho_{S}}[P_{S}(\alpha) = 0]+2\tau\delta \, .
                \end{align*}
                The above shows that for any $\theta \in [\tau\delta, 2\tau\delta]$ and any  $P_{S}\in\cP_2\{V_\rU\}$ we have
                \begin{equation*}
                    \mathbbm{1}_{\theta \in [\xi_{x^S},\xi_{y^S})} \cdot \Pr_{\alpha \in \rho^\theta_S}[P_S(\alpha)=0] \le \left(\Pr_{\alpha \sim \rho_{S}}[P_{S}(\alpha) = 0]+2\tau\delta \right) \cdot \mathbbm{1}_{\theta \in [\xi_{x^S},\xi_{y^S})} \, .
                \end{equation*}
                Further, for every $b\in \{0,1\}$ let $\alpha^{S,b} \in \{0,1\}^S$ be the vector that assigns $y^S$ with $b$ and $x^S,z^S$ with their respective ruling values. Then, we can see that if $P_S(\alpha^{S,b})=1$ for all $b \in \{0,1\}$, then
                \begin{equation*}
                     \mathbbm{1}_{\theta \in [\xi_{x^S},\xi_{y^S})} \cdot \Pr_{\alpha \in \rho^\theta_S}[P_S(\alpha)=0] = 0 \, ,
                \end{equation*}
                which in particular gives
                \begin{align*}
                    \E_{\theta \sim [\tau\delta, 2\tau\delta]}\left[\Delta^\theta_{2,1}\right] & \le \sum_{\substack{P_S \in \cP_2\{V_\rU\} \, : \\ \exists b \in \{0,1\} \text{ s.t. } P_S(\alpha^{S,b})=0 }} \left(\Pr_{\alpha \sim \rho_{S}}[P_{S}(\alpha) = 0]+2\tau\delta \right) \cdot \Pr_{\theta \sim [\tau\delta, 2\tau\delta]}\left[\theta \in [\xi_{x^S},\xi_{y^S})\right] \\
                    & =  \lp_2+\sum_{\substack{P_S \in \cP_2\{V_\rU\} \, : \\ \exists b \in \{0,1\} \text{ s.t. } P_S(\alpha^{S,b})=0 }}2\cdot \left(\xi_{y^S}-\xi_{x^S}\right)
                    \, .
                \end{align*}
                We are then left with the task of upper-bounding $\xi_{y^S}-\xi_{x^S}$. Let us then fix $P_{S}\in\cP_2\{V_\rU\}$ such that either $P_S(\alpha^{S,0})=0$ or $P_S(\alpha^{S,1})=0$. Call $q_S \in \{0,1\}$ the (unique) bit so that $P_S(\alpha^{S,q_S})=0$. Then, by definition of $\xi_{y^S}$ we have
                \begin{align*}
                    \xi_{y^S} & = \min_{b \in \{0,1\}} \Pr_{\alpha\sim \rho_{\{y^S\}}}[\alpha=b] \\
                    & \le \Pr_{\alpha\sim \rho_{\{y^S\}}}[\alpha=q_S] \\
                    & = \Pr_{\alpha\sim \rho_{S}}\left[\alpha_{x^S}=1-\hat{\beta}_{x^S}, \, \alpha_{y^S}=q_S\right] + \Pr_{\alpha\sim \rho_{S}}\left[\alpha_{x^S}=\hat{\beta}_{x^S}, \, \alpha_{y^S}=q_S\right] \\
                    & \le \xi_{x^S} + \Pr_{\alpha\sim \rho_{S}}\left[P_S(\alpha)=0\right] \, .
                \end{align*}
                We can then conclude
                \begin{equation*}
                    \E_{\theta \sim [\tau\delta, 2\tau\delta]}\left[\Delta^\theta_{2,1}\right] \le  \lp_2+\sum_{\substack{P_S \in \cP_2\{V_\rU\} \, : \\ \exists b \in \{0,1\} \text{ s.t. } P_S(\alpha^{S,b})=0 }}2\cdot \left(\xi_{y^S}-\xi_{x^S}\right) \le 3\lp_2 \, .
                \end{equation*}

            \end{proof}

            Note that the only part of the proof of the \Cref{lemma:threshround} that uses randomization in $\theta$ is \Cref{lem:20round}. Let us see how to derandomize this. From \Cref{lem:20round}, we already know that exists a good $\theta\in[\tau\delta, 2\tau\delta]$, for which the thresholding ensures $\Delta^\theta_{2,1}\leq 2\lp_2$. The goal is to have only $O(|V_\rU|)$ many possible values of $\theta$ which we can branch on, ensuring that one of the values would satisfy the required conditions. The idea is simple, on the range from $[0,1]$, we can think of the biases of all the unfixed variables dividing the range into $|V_\rU|+1$ many intervals (this is in fact very crude for us, as we only need the range to be $[\tau\delta, 2\tau\delta]$). Now, taking any threshold in a given interval would round the same set of variables, so we can think of the thresholds as only the boundaries of the intervals. More formally,
            \begin{claim}[Derandomizing thresholding]\label{lem:derand}
                Let $S=\{\tau\delta, 2\tau\delta\}\cup (\{\delta_i\}_{i\in[|V_\rU|]}\cap[\tau\delta, 2\tau\delta])$ where $\delta_{i} = \min_{b\in\{0,1\}}(\Pr_{\alpha\sim\rho_{\{v_i\}}}[\alpha = b])$ for all $v_i\in V_\rU$. Then, there is a $\theta\in S$ such that $\Delta^\theta_{2,1}\leq 2\lp_2$. Moreover, the domain size for $\theta$ is $|S|\leq |V_\rU|+3$.
            \end{claim}
            \begin{proof}
                Without loss of generality, assume that the variables as sorted in increasing order of biases, i.e., $\delta_i < \delta_j$ for all $i<j\in[|V_\rU|]$ (we can assume a strict inequality as for the threshold rounding purposes, as all the variables with the same bias can be thought of as a single variable).  Now, we prove that $\theta_{\text{good}}$ belongs to the set $S$, where $S=\{s_j\}_{j\in[|S|]}:=\{\tau\delta, 2\tau\delta\}\cup (\{\delta_i\}_{i\in[|V_\rU|]}\cap[\tau\delta, 2\tau\delta])$ is the ordered set ($s_i<s_j, \forall i<j$) of potential values for the threshold. Let us assume this is not the case, and $\theta_{\text{good}}\in(s_i, s_{i+1})$ and let $R_{\theta_{\text{good}}}=\{ \forall v_i\in V_\rU, \delta_i \leq \theta_{\text{good}}\}$ be the rounded variables with the threshold as $\theta_{\text{good}}$. Now, note that $\theta' = s_i$ has the same set of rounded variables, as $R_{\theta'}=\{ \forall v_i\in V_\rU, \delta_i \leq \theta'\} = \{ \forall v_i\in V_\rU, \delta_i \leq \theta_{\text{good}}\}$, as there is no variable in $V_\rU$ that has bias in $(s_i, s_{i+1})$. Thus, we can assume $\theta_{\text{good}} = \theta'$.
            \end{proof}
            
            \noindent
            \Cref{lem:20round} and \Cref{lem:derand} establish that there are at most $O(|V_\rU|)$ many possible values for $\theta\in[\tau\delta, 2\tau\delta]$ such that on thresholding with $\theta$ we get $\Delta^\theta_{2,1}\leq 2\lp_2$. Moreover, these values can be found in time $\poly(|V_\rU|)$. Let us then fix $\theta$ be any such that threshold, and calculate the overall increase in the aggregate value. We note that:
            \begin{itemize}
            \item $\lp'_3 \leq \lp_3$ by definition;
            \item $\lp'_2 = \Delta_{2, 2} + \Delta_{3, 2}  \leq \lp_2 + \lp_3+6\tau\delta\cdot\binom{V_\rU}{3}$ by \Cref{lem:3round};
            \item $\lp'_1 = \Delta_{1, 1} + \Delta_{2, 1} + \Delta_{3, 1}  \leq 
            \lp_1 + 3\lp_2 +\lp_3 +6\tau\delta\cdot\binom{V_\rU}{3} $ by \Cref{lem:3round} and \Cref{lem:derand};
            \item $\lp'_0 = \Delta_{0, 0} + \Delta_{1, 0} + \Delta_{2, 0} + \Delta_{3, 0}  \leq 
            \lp_0 + 
            2\lp_1 + 2\lp_2 + \lp_3+6\tau\delta\cdot\binom{V_\rU}{3}$ by \Cref{lem:10}, \Cref{lem:11}, and \Cref{lem:3round}. 
            \end{itemize}
            Combining the above three quantities, we get that
            \begin{align*}
                A'(\rho') - A(\rho)
                & \leq \left(\log n +1\right)\cdot 3\lp_2 + 2\lp_1 + \left(\lp_3+6\tau\delta\binom{V_\rU}{3}\right)(\log^2 n + \log n + 1) \\
                & \le 6 \cdot \left(\log^2 n \cdot \lp_3 + \log n \cdot \lp_2 + \lp_1 \right) + 12\tau\delta \binom{V_\rU}{3} \log^2 n \\
                & \leq \frac{6}{\log n}A(\rho)+ 12\tau\delta \binom{V_\rU}{3} \log^2 n \, ,
            \end{align*}
            which shows that there is a threshold $\theta \in [\tau\delta,2\tau\delta]$ such that $\theta \in \Theta^{V_\rU}_\tau(\rho,\delta)$.
        \end{proof}

	\section{Polynomial time algorithm for complete $k$-CSPs}
\label{sec:kcsp}

In this section, we show a simple algorithm that, given any complete $n$-variable $k$-CSP instance $\cI = (V, \cP)$ on the boolean alphabet $\{0,1\}$, and decides if there is a satisfying assignment to $\cI$ in $n^{O(k)}$ time. This improves the algorithm of~\cite{anand2025min}, which showed a quasi-polynomial time algorithm for any fixed $k$.

The quasi-polynomial time of~\cite{anand2025min} uses a bound on the number of satisfying assignments of any boolean CSP, proved by the same authors. We use the same tool, which can be stated as follows.

\begin{lemma}[Lemma 3.1 of~\cite{anand2025min}]\label{lemma:vc}
Let $\cI = (V,\cP)$ be a complete $k$-CSP. Then, the number of satisfying assignments to $\cI$ is at most $O(|V|^{k-1})$.
\end{lemma}
\noindent
Our algorithm is very simple. It chooses an ordering $v_1,v_2, \dots, v_n$ of $V$, and it keeps track of all satisfying assignments for the sub-instance induced on the first $i$ variables for $i = k, k + 1, \dots, n$. Note that the instance induced on the first $i$ variables is a complete CSP with $i$ variables, hence for any $i \in [n]$ the number of such assignments is at most $O(i^{k-1})$ by~\Cref{lemma:vc}. At the $i$-th step, we simply try to extend each satisfying assignment by trying $x_i = 0$ and $x_i = 1$. At the end, it suffices to test if there is an satisfying assignment left.

\begin{algorithm}[H]\label{alg:decisioncsp}
\caption{\textsc{decide-csp}{}(\cI)}
\begin{algorithmic}[1]
\Require complete $k$-CSP $\cI=(V,\cP)$ with $n=|V|$.
\State Arbitrarily order the variables in $V$ as $v_1, v_2, \ldots, v_n$.
\State Let $\mathcal{S}_{k-1} = \{0,1\}^{k-1}$
\For {$i = k, k + 1 \ldots n$}
\State Let $\mathcal{S}_i = \emptyset$.
\State Let $V' = \{v_1,\dots,v_{i}\}$.
\For {$\alpha \in \mathcal{S}_{i-1}$}
\For{$b = 0, 1$}
\State Let $\alpha' \in \{0,1\}^{V'}$ be defined as $\alpha'_{v_i}=b$ and $\alpha'_{v_j} = \alpha_{v_j}$ for each $j \in [i-1]$.
\If{$\val(\alpha',\cI[V'])=1$}
\State Add $\alpha'$ to $\mathcal{S}_i$.
\EndIf
\EndFor
\EndFor
\EndFor
\If{$\mathcal{S}_n \neq \emptyset$}
\State \Return \textsc{yes}
\Else
\State \Return \textsc{no}
\EndIf
\end{algorithmic}
\end{algorithm}

\noindent
The correctness of the algorithm is immediate. At every stage, the number of satisfying assignments, and hence the size of $\mathcal{S}$, is at most $n^{k-1}$. Hence the algorithm runs in $n^{O(k)}$ time.
	
	\addcontentsline{toc}{section}{References}
	\bibliographystyle{alpha}
	\bibliography{refs}
	
	\appendix
	 \section{Deferred proofs from Section~\ref{sec:roundingalgo}}
 
\subsection{Proof of \Cref{lem:bruteforceok}}
\label{sec:edgecases1}

In this section we prove the correctness of the base case of the recursion we the execution of \Cref{alg:roundingalgo} reaches line~\ref{line:bruteforce}. The claim is restated below for convenience of the reader.

\bruteforce*

\begin{proof}
Notice that line~\ref{line:bruteforce} is reached after at most $100 \log n$ stages of conditioning and thresholding, by \Cref{lem:returnpoints}. For each conditioning step we lose at most $r + 2t \leq 2r = 2(\log n)^{100K}$ degrees from the Sherali-Adams relaxation. Hence, at this point we have a pseudodistribution $\mu \in \cD(V, d)$ where $d = K^2\log^{c} n - 100 \log n \cdot 2(\log n)^{100K}\geq K^2 \log^{100} n \geq |V_\rU|$.

This in turn means that the Sherali-Adams solution $\mu$ is a distribution over integral solutions to the unfixed variables $V_\rU$. Hence in particular, it follows that there exists an integral assignment $\beta_\rU$ to the variables in $V_\rU$ such that the combined assignment $\beta$, which assigns $\beta_\rU$ to variables of $V_\rU$, and the already fixed assignment $\alpha_{V_\rF}$ to the variables $V_\rF$, satisfies $\val(\beta) \leq \val(\mu)$, and such an assignment would be found by brute force.
\end{proof}

\subsection{Proof of \Cref{lem:highlyunsatok}}
\label{sec:edgecases2}

In this section we show that when the sub-instance induced by the unfixed variables has large objective value, it means that these constraints are not important and we can then focus on solving the \mintwo instance given by the constraints with one or two fixed variables. The claim is restated below for convenience of the reader.

\twosatcase*

\begin{proof}Notice that line~\ref{line:2satinstance} is reached after at most $100 \log n$ stages of conditioning and thresholding, by \Cref{lem:returnpoints}. For each conditioning step we lose at most $r + 2t \leq 2r = 2(\log n)^{100K}$ degrees from the Sherali-Adams relaxation. Hence, at this point we have a pseudodistribution $\mu \in \cD(V, d)$ where $d = K^2\log^{c} n - 100 \log n \cdot 2(\log n)^{100K}\geq K^2 \log^{100} n \ge 3$.

We now observe that  for all $\beta \in \{0,1\}^{V_\rU}$, since $\val(\cI[V_\rU],\mu[V_\rU]) > 1/(10\tau)$, the average value of unsatisfied constraints can be bounded as
\begin{equation*}
    \val(\cI[V_\rU],\beta) \le 1 \le 10\tau \cdot \val(\cI[V_\rU],\mu[V_\rU]) \, .
\end{equation*}
Now, \Cref{alg:roundingalgo} constructs a \mintwo instance $\cI_2\langle V_\rU,\alpha\rangle=(V_\rU, \cP_2\langle V_\rU,\alpha\rangle)$ by ignoring constraints which have all their variables from $V_\rU$ and those constraints which have all their variables from $V_\rF=V\setminus V_\rU$. Hence, every constraint in $\cI_2\langle V_\rU,\alpha\rangle$ involves $1$ or $2$ variables from $V_{\rU}$. We recall that since all $P_S \in \cP$ are $\nae$, every $P'_S \in \cP_2\langle V_\rU, \alpha\rangle$ is a \twosat or \nsat{1} clause.

Let $V_{\rU}^- = \{\neg v : v \in V_{\rU}\}$ be the set of negative literals.
For convenience, we will represent this instance as $\cI'=(V_{\rU}, \calC)$ where the set $\calC$ is the multiset of clauses in the instance, corresponding to the constraints $\cP_2\langle V_\rU,\alpha\rangle$. Each \twosat clause in $\calC$ is the disjunction of two literals $\ell_1 \lor \ell_2$ with $\ell_1,\ell_2 \in V_\rU \cup V_\rU^-$. We also allow ourselves to alternatively write \twosat clauses as follows: given $\ell_1\lor \ell_2$ where $\ell_1,\ell_2 \in V_\rU \cup V_\rU^-$ corresponding to variables $v,w \in V_\rU$ respectively, we identify $\ell_1 \lor \ell_2$ with a tuple $(v,w,p_1,p_2)$ where $p_1,q_2 : \{0,1\} \rightarrow \{0,1\}$ are bijections describing the literal pattern of $\ell_1$ and $\ell_2$ respectively. More precisely, each of $p_1$ and $p_2$ is identical to one of the mappings $\id: b \mapsto b$ or $\inv: b \mapsto 1-b$. For example, if we have a clause $\neg v \lor w$ we might also write it as $(v,w,\inv,\id)$. Note that this handles \nsat{1} clauses too, since $v$ is allowed to equal $w$. Let us further define
$$\val(\cI', \mu) = \ex_{ (v,w,p_1,p_2) \sim \calC}\left[\Pr_{\sigma \sim \mu_{\{v,w\}}}[p_1(\sigma_v) = 0, \, p_2(\sigma_w) = 0]\right]$$
to be the average probability with which a clause is unsatisfied. We note that
$$|\calC|\val(\cI', \mu) \leq \lp_2^{V_{\rU}}(\mu) + \lp_1^{V_{\rU}}(\mu) \, , $$
since each clause has either $1$ or $2$ variables from $V_{\rU}$, and the remaining variables from $V \setminus V_\rU$ (recall that \smash{$\lp^{V_\rU}_i(\mu)$} is the sum of the values of the clauses with exactly $i$ variables from $V_\rU$). 

Our goal will be to show that the Sherali-Adams solution $\mu$ at this point gives a feasible solution for a standard LP relaxation for the \textsc{Min-$2$-SAT}{} instance. If we show this, then by standard results in the literature, we know that the gap for this standard LP relaxation is at most $O(\log^2 n)$. In turn, this must imply that running the rounding algorithm certifying this gap would give an assignment $\alpha' \in \{0,1\}^{V_\rU}$ that violates at most $ O(\log^ 2 n)(\lp_2^{V_\rU}(\mu) + \lp_1^{V_\rU}(\mu))$ clauses in $\cI'$. Putting it together with $\alpha$ by including the constraints with all three variables from $V_\rU$, and the constraints with all three variables from $V \setminus V_\rU$, we get an assignment $\alpha^* \in \{0,1\}^V$ such that the total number of constraints it violates must be at most
$$10\tau \lp_3^{V_\rU}(\mu) + O(\log^2 n)(\lp^{V_\rU}_2(\mu) + \lp^{V_\rU}_1(\mu)) + \lp^{V_\rU}_0(\mu) \leq O(\log n)\, A^{V_\rU}(\mu) \, ,$$
which will finish the proof.

We now restrict our attention to showing that the Sherali-Adams solution $\mu$ is feasible for the standard LP relaxation for our \mintwo instance (see for example~\cite{klein1997approximation}), stated below for convenience. This relaxation solves for a metric $d: V_\rU\cup V_\rU^{-} \rightarrow \R_{\ge 0}$ over the set of literals  $V_\rU \cup V_\rU^-$.

\begin{align*}
    \text{minimize} \quad & \sum_{(\ell_1\lor\ell_2) \in \calC} \frac{1}{2}(d(\neg\ell_1,\ell_2) + d(\neg\ell_2,\ell_1)) \\
    \text{subject to} \quad &
    d(v, \neg v) + d(\neg v , v) \geq 1, \quad \forall v \in V_{\rU} \quad \quad \quad  \quad \quad \quad \text{(LP-2-SAT)} \\
    & d(\ell_1,\ell_2) \geq 0 \quad \forall \ell_1,\ell_2 \in  V_\rU \cup V_\rU^{-} \\
    & d(\ell_1,\ell_3) \leq d(\ell_1,\ell_2) + d(\ell_2,\ell_3) \quad \forall \ell_1,\ell_2,\ell_3 \in  V_\rU \cup V_\rU^{-}
\end{align*}

\noindent
Notice that this is indeed a relaxation: given any assignment, for each violated clause $(\ell_1 \lor \ell_2) \in \calC$ we set $d(\neg \ell,\ell) = d(\neg \ell, \ell) = 1$, for every satisfied clause $(\ell_1 \lor \ell_2) \in \calC$ we set $d(\neg \ell,\ell) = d(\neg \ell, \ell) = 0$, and for every other $\ell_1,\ell_2 \in  V_\rU \cup V_\rU^{-}$ we obtain $d(\ell_1,\ell_2)$ by metric completion. Formally, $ d(\ell_1,\ell_2)$ will be the shortest length $d(\ell_1,z_1) + d(z_1,z_2) +\ldots +d(z_t,\ell_2)$ among all choices of literals $z_1,\dots,z_t \in V_{\rU} \cup V_{\rU}^-$ such that $(\neg \ell_1 \lor z_1), (\neg z_1 \lor z_2),\dots, (\neg z_t, \ell_2)\in \calC$. In other words, we turn each clause $\ell_1 \lor \ell_2$ into two implications $\neg \ell_1 \implies \ell_2$, $\neg \ell_2 \implies \ell_1$ and assign the lengths $d({\neg \ell_1},{\ell_2}), d({\neg \ell_2},{\ell_1})$ to be $1$ if and only if the corresponding implications are unsatisfied, and assign the other distances as the shortest path distance in the graph whose edges are formed by these implications. We have the following standard result from classical approximation algorithms literature.

\begin{theorem}[\cite{klein1997approximation}]\label{thm:2satrounding}
There is a rounding algorithm that, given any feasible solution to {\upshape{(LP-2-SAT)}} of objective value $\phi \in \mathbb{R}_{\ge 0}$, outputs in polynomial time an assignment for which there are at most $O(\log^2 n) \phi$ unsatisfied clauses.
\end{theorem}

\noindent
Now we show that given a degree-$d$ Sherali-Adams solution $\rho$ with $d \geq 3$, we can obtain a feasible solution to (LP-2-SAT) whose value is the same as the value of the Sherali-Adams relaxation. We remark that this is a standard fact, but we prove it here for completeness.

\begin{claim}\label{lemma:convertLP}
Let $d \ge 3$ be an integer and let $\mu \in \cD(V,d)$. Then, one can construct a feasible solution to {\upshape{(LP-2-SAT)}} with objective value at most $ |\calC|\val(\cI',\mu)$.\footnote{We note that $|\calC|$ is a scaling constant since $\val$ is defined as an average, whereas the objective of (LP-2-SAT) is defined as a sum.} 
\end{claim}

\begin{proof}
The construction is very straightforward. Then, for each $\ell_1,\ell_2 \in V_{\rU} \cup V_{\rU}^-$, letting $(v,w,p_1,p_2)$ be the tuple representation of $\ell_1 \lor \ell_2$, we set \smash{$d(\ell_1,\ell_2) = \Pr_{\sigma \in \mu_{\{v,w\}}}[p_1(\sigma_v)=1, \, p_2(\sigma_w)=0]$}. In words, the length of the edge $\ell_1,\ell_2$ is the probability that the implication $\ell_1 \implies \ell_2$ is falsified in the pseudodistribution $\mu$. It remains to verify that this solution is feasible. Clearly all the $d(\ell_1,\ell_2)$ are non-negative, and note that for each $v \in V_{\rU}$ we have
$d(v,\neg v) = \Pr_{\sigma \in \mu_{\{v\}}}[\sigma=1]$ and $d(\neg v,v) = \Pr_{\sigma \in \mu_{\{v\}}}[\sigma=0]$, which means $d(\neg v,v)+ d(v,\neg v)=\Pr_{\sigma \in \mu_{\{v\}}}[\sigma=0]+\Pr_{\sigma \in \mu_{\{v\}}}[\sigma=1]=1$. Next, we check the feasibility of the triangle inequality. Let $\ell_1,\ell_2,\ell_3 \in  V_\rU \cup V_\rU^{-}$, and let $(u,v,p_1,p_2)$, $(v,w,p_2,p_3)$, $(u,w,p_1,p_3)$ be the tuple representation of $\ell_1\lor\ell_2$, $\ell_2\lor\ell_3$, $\ell_1\lor\ell_3$, respectively. Then
\begin{align*}
& d(\ell_1, \ell_3) \\
= & \Pr_{\sigma \sim \mu_{\{u,w\}}} [p_1(\sigma_u) = 1, \, p_3(\sigma_w) = 0] \\ 
= & \Pr_{\sigma \sim \rho_{\{u,v,w\}}}[p_1(\sigma_u) = 1, \, p_2(\sigma_v) = 0, \, p_3(\sigma_w) = 0] + \Pr_{\sigma \sim \rho_{\{u,v,w\}}}[p_1(\sigma_u) = 1, \, p_2(\sigma_v) = 1, \, p_3(\sigma_w) = 0] \\
\leq & \Pr_{\sigma \sim \mu_{\{u,v\}}} [p_1(\sigma_u) = 1, \, p_2(\sigma_v) = 0]+ \Pr_{\sigma \sim \mu_{\{v,w\}}} [p_2(\sigma_v) = 1, \, p_3(\sigma_w) = 0] \\
\leq & d(\ell_1, \ell_2) + d(\ell_2, \ell_3)
\end{align*}
as desired. The LP value of this solution $\frac{1}{2}\sum_{(\ell_1 \lor \ell_2) \in \calC} d({\neg \ell_1},{\ell_2}) + d({\neg \ell_2},{\ell_1})$ is upper-bounded as
\begin{align*}
& \frac{1}{2}\sum_{(\ell_1 \lor \ell_2)=(v,w,p_1,p_2) \in \calC} \Pr_{\sigma \sim \mu_{\{v,w\}}}[\inv(p_1(\sigma_v))=1, \, p_2(\sigma_w))=0]+ \Pr_{\sigma \sim \mu_{\{v,w\}}}[\inv(p_2(\sigma_w))=1, \, p_1(\sigma_1))=0]  \\
= &  \frac{1}{2}\sum_{i,j \in C} 2\Pr_{\sigma \sim \mu_{\{v,w\}}}[p_1(\sigma_v))=0, \, p_2(\sigma_w))=0] \\
\leq & |\calC|\val(\cI',\mu) \, ,
\end{align*}
where the last inequality follows from the definition of $\val(\cI',\mu)$. 
\end{proof}
\end{proof}

	\section{Hardness of \minnaethree\ on dense instances}
We show that solving \minnaethree exactly on dense instances is almost as hard as on general instances (which is known to be \nph \cite{sch78}).

\begin{claim}
For any small enough constant $0<\eps <1/1000$, there is an approximation-preserving polynomial time reduction from a general instance of \minnaethree to a dense instance of \minnaethree whose constraint hypergraph has at least $(1 - \eps) \binom{n}{3}$ constraints and every variable appears in at least $(1-\eps)\binom{n}{2}$ constraints. 
\label{claim:hardness-dense}
\end{claim}
\begin{proof}
Given a general instance of \minnaethree with variable set $V_0$ with $n_0 = |V_0|$ and constraints $\calC$, add $|V_d|= O(n_0 / \eps)$ dummy variables. The total number of variables becomes $n=n_0+|V_d|$. First, we add $\binom{|V_d|}{3}$ constraints on $V_d$, to make sure that the all-true assignment satisfies them. To ensure this, for every three variables $v_1, v_2, v_3\in V_d$, we add a \nae clause $(v_1, v_2, \neg v_3)$. Now, we add the constraints between the dummy variables $V_d$ and $V_0$. For all pairs of variables $v_1, v_2\in V_d$ and variables $v\in V_0$, add a \nae constraint $(v, v_1, \neg v_2)$. Output the original instance combined with the dummy variables and constraints. The number of constraints is at least \smash{$\binom{|V_d|}{3} +  |V_0|\cdot \binom{|V_d|}{2}= O(n_0^3/\eps^3)$}, which is at least \smash{$(1-\eps)\binom{n}{3}$} for a small enough $\eps$. Moreover, the number of clauses each variable appears in, is at least $\binom{|V_d|}{2}\geq (1-\eps) \binom{n}{2}$ for a small enough $\eps$.

For any assignment of the original variables, one can easily satisfy all the dummy constraints by setting all the dummy variables to true. In the other direction, for any assignment of the new instance, changing all dummy variables to true will only satisfy more constraints, so it will possibly violate only the original constraints. 

Therefore, the optimal values of the two instances are the same.
\end{proof}
	
\end{document}